\documentclass[smallextended]{svjour3}
\usepackage{graphicx, amsmath, amssymb}
\usepackage[caption=false]{subfig}
\usepackage{xcolor}

\allowdisplaybreaks[1] 

\newcommand{\braket}[2]{{\left\langle #1 \middle| #2 \right\rangle}}

\newcommand{\ket}[1]{{\left| #1 \right\rangle}}

\newcommand{\fref}[1]{Fig.~\ref{#1}}
\newcommand{\sref}[1]{Section~\ref{#1}}
\newcommand{\norm}[1]{\Vert #1 \Vert}

\usepackage{ulem}

\newtheorem{observation}{Observation}
\renewcommand{\qed}{\hfill \ensuremath{\Box}}

\journalname{Quantum Inf Process}

\begin{document}

\title{Simplifying Continuous-Time Quantum Walks on Dynamic Graphs}

\author{Rebekah Herrman \and Thomas G.~Wong}

\authorrunning{R.~Herrman \and T.~G.~Wong}

\institute{R.~Herrman \at
	Department of Industrial and Systems Engineering, The University of Tennessee, Knoxville, Tennessee 37996, USA \\
	\email{rherrma2@tennessee.edu}
	\and
	T.~G.~Wong \at
	Department of Physics, Creighton University, 2500 California Plaza, Omaha, Nebraska 68178, USA \\
	\email{thomaswong@creighton.edu}
}

\date{Received: date / Accepted: date}

\maketitle

\begin{abstract}
	A continuous-time quantum walk on a dynamic graph evolves by Schr\"odinger's equation with a sequence of Hamiltonians encoding the edges of the graph. This process is universal for quantum computing, but in general, the dynamic graph that implements a quantum circuit can be quite complicated. In this paper, we give six scenarios under which a dynamic graph can be simplified, and they exploit commuting graphs, identical graphs, perfect state transfer, complementary graphs, isolated vertices, and uniform mixing on the hypercube. As examples, we simplify dynamic graphs, in some instances allowing single-qubit gates to be implemented in parallel.
\end{abstract}


\section{Introduction}

A continuous-time quantum walk is the quantum version of a continuous-time random walk, where the walker hops to adjacent vertices on a graph by evolving by Schr\"odinger's equation
\begin{equation}
	\label{eq:Schrodinger}
	i \frac{d\ket{\psi}}{dt} = H \ket{\psi},
\end{equation} 
where we have set $\hbar = 1$. Continuous-time quantum walks were first introduced by Farhi and Gutmann \cite{FG1998a} as a means of traversing decision trees. In some cases, a classical random walk would take exponential time to traverse the decision tree, whereas the quantum walk would only take polynomial time, although faster classical algorithms existed. Subsequently, Childs \cite{Childs2003} constructed a graph by gluing together two binary trees using a random cycle and showed that a continuous-time quantum walk traversed it exponentially faster than any classical algorithm, relative to an oracle, giving the first exponential speedup by quantum walk. Continuous-time quantum walks have also been used for searching a graph for a marked node \cite{CG2004}, perfect state transfer \cite{Christandl2004}, and evaluating boolean formulas \cite{FGG2008}. They have also been shown to be universal for quantum computation \cite{Childs2009} using a scattering approach, where each computational basis state corresponds to a rail of vertices to walk on. Quantum walks that evolve in discrete-time also exist \cite{Aharonov2001}, but in this paper, quantum walks will henceforth refer to their continuous-time versions.

\begin{figure}
\begin{center}
	\includegraphics{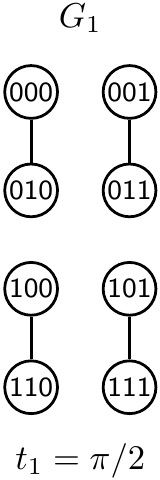} \quad \quad \quad
	\includegraphics{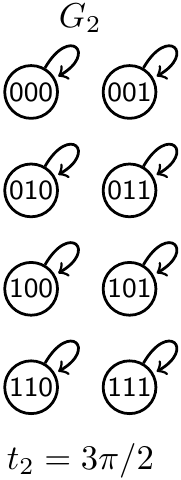}
	\caption{\label{fig:IXI}A dynamic graph of length 2 on eight vertices that implements $I \otimes X \otimes I$.}
\end{center}
\end{figure}

In typical studies of quantum walks, the graph on which the walker moves is static. Some limited research has been done, however, on quantum walks on graphs whose edges change at discrete times, i.e., sequences of graphs. The first was \cite{Underwood2010}, who encoded each computational basis state as a ``rail'' of vertices. Then, \cite{coutinho2019discretization} studied a discretization of continuous-time quantum walks in which tessellations were used to define Hamiltonians which were used to evolve the quantum state. Later \cite{Chakraborty2017}, encoded each computational basis state by a single vertex, and they considered quantum walks that evolved for the same amount of time on each graph of the sequence. Recently, \cite{HH2019} generalized this to permit the walk on each graph to occur for different amounts of time, and these were called dynamic graphs. Formally, a dynamic graph is defined as $\mathcal{G} = \{(G_i, t_i)\}$ for $i \in \mathbb{Z}^+$, where $G_i= (V_i, E_i)$ is a graph with vertex set $V_i$ and edge set $E_i$, and $t_i$ is the amount of time that the walk occurs on graph $G_i$. In this paper, we call a quantum walk on a dynamic graph a \emph{dynamic quantum walk}. Then, the dynamic quantum walk occurs on graph $G_1$ from time $t \in (0,t_1)$, on graph $G_2$ from time $t \in (t_1, t_1+t_2)$, on graph $G_3$ from time $t \in (t_1+t_2,t_1+t_2+t_3)$, and so forth. For example, a dynamic graph of length 2 is shown in \fref{fig:IXI}, and the eight vertices are the computational basis states of three qubits. The dynamic quantum walk evolves on $G_1$ for time $t_1 = \pi/2$ and then on $G_2$ for an additional time of $t_2 = 3\pi/2$, for a total evolution time of $2\pi$. If the initial state of the walker is $\ket{\psi(0)} = a_0 \ket{000} + a_1 \ket{001} + \dots + a_7 \ket{111}$, and the Hamiltonian is equal to the adjacency matrix of the graph (i.e., $H = A$, where $A_{ij} = 1$ if vertices $i$ and $j$ are adjacent, and $A_{ij} = 0$ otherwise), then after the first graph, the state of the walker is
\begin{align*}
    e^{-iA_1t_1} \ket{\psi(0)}
        &= -i \big( a_2 \ket{000} + a_3 \ket{001} + a_0 \ket{010} + a_1 \ket{011} \\
        &\quad\quad\quad + a_6 \ket{100} + a_7 \ket{101} + a_4 \ket{110} + a_5 \ket{111} \big),
\end{align*}
and after the second graph, the state of the walker is
\begin{align*}
    e^{-iA_2t_2} e^{-iA_1 t_2} \ket{\psi(0)}
        &= a_2 \ket{000} + a_3 \ket{001} + a_0 \ket{010} + a_1 \ket{011} \\
        &\quad + a_6 \ket{100} + a_7 \ket{101} + a_4 \ket{110} + a_5 \ket{111}.
\end{align*}
This final state is exactly what would be obtained by applying $I \otimes X \otimes I$ to the initial state, where $X$ is the Pauli-X gate.

Ref.~\cite{HH2019} showed how various quantum gates, such as the Pauli gates, can be implemented using dynamic quantum walks, such as our previous example in \fref{fig:IXI} that applies the $X$-gate to the middle qubit. Ref.~\cite{HH2019} also showed how to implement the universal gate set consisting of the Hadamard gate, $T$-gate, and CNOT gate. Since this approach uses one vertex rather than one rail for each computational basis state, this universality result uses a smaller Hilbert space than \cite{Childs2009} at the expense of dynamically changing the edges. Note while \cite{HH2019} did use ancillary vertices for some graphs, they were shown to be unnecessary in \cite{Wong33} by allowing isolated vertices to be loopless, so a quantum computation in an $N$-dimensional Hilbert space uses exactly $N$ vertices when implemented using a quantum walk on a dynamic graph.

\begin{figure}
\begin{center}
	\subfloat[] {
		\label{fig:layers}
		\includegraphics{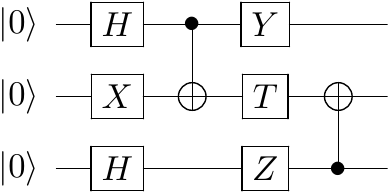}
	}
	
	\subfloat[] {
		\label{fig:sequential}
		\includegraphics{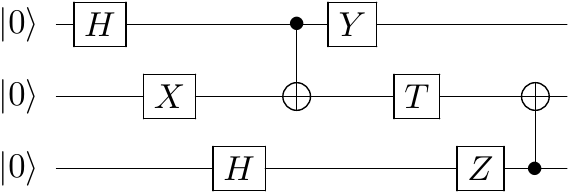}
	}
	\caption{(a) A quantum circuit of three qubits with four alternating layers of one- and two-qubit gates. (b) The same circuit, but expanded to show the order of gates performed by the quantum walk.}
\end{center}
\end{figure}

In \cite{Wong33}, a dynamic quantum walk was used to simulate the circuit in \fref{fig:layers}, which alternates between single-qubit gates and two-qubit gates in a fashion similar to the first quantum computational supremacy experiment \cite{Google2019}. Although the single-qubit gates in \fref{fig:layers} are drawn in such a way as to suggest they are applied in parallel, the dynamic quantum walk actually implemented them sequentially, as shown in \fref{fig:sequential}. This motivates the following question:
\begin{question}
	Is there a way to apply single-qubit gates in parallel using quantum walks on dynamic graphs?
\end{question}
\noindent More broadly,
\begin{question}
	How can a dynamic graph $\mathcal{G}$ be simplified by reducing the number of graphs or evolving for less time?
\end{question}
This paper addresses these two questions. In \sref{sec:timeandhamiltonian}, to have a consistent measure for the amount of time that a dynamic quantum walk takes, we discuss rescaling the Hamiltonian by the spectral norm of the adjacency matrix. In \sref{sec:rules}, we give six characteristics of dynamic graphs that allow them to be simplified so that the resulting dynamic quantum walk either contains fewer graphs or evolves shorter in time, and we give examples of them. The simplifications exploit commuting graphs, identical graphs, perfect state transfer, complementary graphs, isolated vertices, and uniform mixing on the hypercube. Examples are shown for each of these. In \sref{sec:simplification}, we use these observations to simplify the dynamic graph that implements \fref{fig:sequential}. Finally, in \sref{sec:conclusion}, we summarize our results and discuss future work.


\section{Time and the Hamiltonian}\label{sec:timeandhamiltonian}

In previous work on dynamic quantum walks \cite{HH2019,Wong33}, the Hamiltonian was $H = A$. Then, by solving Schr\"odinger's equation \eqref{eq:Schrodinger}, the time-evolution under each graph was $\ket{\psi(t)} = e^{-i A t} \ket{\psi(0)}$. This may not be the best Hamiltonian to use, however, because the jumping rate (amplitude per time) can vary for different graphs, leading to inconsistent measures of time. More precisely, the adjacency matrix of different graphs can have different spectral norms, where the spectral norm of $A$, which we denote $\norm{A}$, is defined as
\[ \norm{A} = \sqrt{\text{largest eigenvalue of $A^\dagger A$}}, \]
where $A^\dagger$ is the conjugate transpose of $A$. Or, since $A$ is symmetric,
\[ \norm{A} = \text{largest eigenvalue of $A$ in absolute value}. \]
For example, the adjacency matrix of the cycle of length 4 (denoted $C_4$) has a spectral norm of $2$, whereas the adjacency matrix of the path graph of length 2 (denoted $P_2$) has a spectral norm of $1$. Then, a quantum walk on $C_4$ has twice the jumping rate of a quantum walk on $P_2$, so it is walking twice as quickly. For a fair comparison, we should multiply $C_4$'s evolution time by two when comparing it to a walk on $P_2$. For example, when $H = A$, perfect state transfer between opposite corners of $C_4$ occurs at $t = \pi/2$, and it occurs on $P_2$ also when $t = \pi/2$. One might say that perfect state transfer takes the same amount of time on each graph, but for a more accurate comparison, $C_4$'s time should be $\pi$. Another way to understand this is through energy. Since the Hamiltonian is the operator corresponding to the total energy of the system, walking on $C_4$ uses twice as much energy as evolving on $P_2$. To make a fair comparison between quantum walks on different graphs, the energy usage should be consistent, so the evolution time of $C_4$ should be doubled.

Rather than rescaling time for each graph, we can have a consistent notion of time by rescaling the Hamiltonian by dividing it by the spectral norm of the adjacency matrix. That is, for the remainder of this paper, we will use the Hamiltonian
\begin{equation}
	\label{eq:Hamiltonian}
	H = \frac{A}{\norm{A}}.
\end{equation}
With this Hamiltonian \eqref{eq:Hamiltonian}, the system evolves under each graph by
\begin{equation}
	\label{eq:evolution}
	\ket{\psi(t)} = e^{-i A t/\norm{A}} \ket{\psi(0)}. 
\end{equation}
Now, the evolution time of different graphs can be directly compared to each other. For example, with this Hamiltonian, perfect state transfer occurs between opposite corners of $C_4$ when $t = \pi$ and on $P_2$ when $t = \pi/2$. Such proper scaling of the Hamiltonian appears in various prior works on quantum walks, such as \cite{moore2002quantum}, where a quantum walk on the $n$-dimensional hypercube is considered, and the Hamiltonian is equal to the adjacency matrix divided by $\norm{A} = n$. Or in \cite{FG1998b,CG2004}, a quantum walk version of Grover's algorithm is presented, and the $O(\sqrt{N})$ optimality of the algorithm relies on constant energy usage.

With the evolution given in \eqref{eq:evolution}, evolution times can be taken modulo the period of the walk. We can calculate the period from the eigenvalues of $A$, which we denote $\lambda_1, \dots, \lambda_N$, and let us denote the corresponding eigenvectors $\ket{\psi_1}, \dots, \ket{\psi_N}$. Expressing the initial state of the quantum walk as a superposition over the eigenvectors, we get
\[ \ket{\psi(0)} = \alpha_1 \ket{\psi_1} + \dots + \alpha_N \ket{\psi_N}. \]
Then, according to \eqref{eq:evolution}, the walk evolves to
\[ e^{-iAt} \ket{\psi(0)} = \alpha_1 e^{-i\lambda_1t/\norm{A}} \ket{\psi_1} + \dots + \alpha_N e^{-i\lambda_Nt/\norm{A}} \ket{\psi_N}. \]
If $\lambda_i \ne 0$, then $e^{-i\lambda_it/\norm{A}}$ is periodic with period $2\pi\norm{A}/\lambda_i$. If $\lambda_i = 0$, then the exponential is 1. Thus, the period of the quantum walk is, if it exists, the least common multiple of $2\pi\norm{A}/\lambda_i$ for all nonzero $\lambda_i$, i.e.,
\[ \text{period} = \text{lcm} \left\{ \frac{2\pi\norm{A}}{\lambda_i} : \lambda_i \ne 0 \right\}. \]
If the least common multiple does not exist, we take the period to be infinite. Then, regardless whether the period is finite or infinite, we can take the evolution time on the graph modulo the period. In the present context of implementing quantum gates using dynamic quantum walks, however, all the previously developed graphs have finite periods \cite{Wong33}, as are all the graphs in this paper. Finally, if $\norm{A} = 1$ and every $\lambda_i$ is in $\{0,1,-1\}$, we can take the evolution time modulo $2\pi$.


\section{Combining Dynamic Quantum Walks}\label{sec:rules}

In this section, we give six observations by which dynamic quantum walks can be simplified, along with examples. We assume each dynamic graph in $\mathcal{G}$ has the same size vertex set. We also write the initial state of each dynamic graph as the superposition $\sum_{i \in V} c_i\ket{i}$, where $\ket{i}$ is the basis state corresponding to vertex $i$ written in binary. From this initial state, we can work through how the state changes under the dynamic quantum walk.


\subsection{Swapping Commuting Sequential Graphs}

The first observation gives a method to swap commuting graphs, meaning their adjacency matrices commute.

\begin{observation}\label{obs:commute}
	If sequential graphs commute, their order in the dynamic graph can be swapped.
\end{observation}

\begin{proof}
	Suppose $G_\ell$ and $G_{\ell+1}$ commute. Then 
	\[ e^{-i A_{\ell+1} t_{\ell+1}/\norm{A_{\ell+1}}}e^{-i A_\ell t_\ell/\norm{A_\ell}} = e^{-i A_\ell t_\ell/\norm{A_\ell}}e^{-i A_{\ell+1} t_{\ell+1}/\norm{A_{\ell+1}}}. \]
	So by definition, so the order of the graphs in the sequence can be swapped.\qed
\end{proof}

As a special case, note that if the only edges of a graph are self-loops on every vertex, the adjacency matrix is the identity matrix, which commutes with all matrices. 

\begin{example}
	Here, we will consider the $X \otimes X$ gate, which acts on computational basis states as
	\begin{align*}
		(X \otimes X) \ket{00} &= \ket{11}, \\
		(X \otimes X) \ket{01} &= \ket{10}, \\
		(X \otimes X) \ket{10} &= \ket{01}, \\
		(X \otimes X) \ket{11} &= \ket{00}.
	\end{align*}
	In \fref{fig:fullxx}, $X \otimes X$ is implemented using a dynamic quantum walk by sequentially implementing $X \otimes I$ followed by $I \otimes X$ using the results from \cite{HH2019,Wong33}. This acts on the initial state via 
	\begin{align*}
		&c_0 \ket{00} + c_1 \ket{01} + c_2 \ket{10} + c_3 \ket{11}\\
		&\quad \xrightarrow{G_1} -i\big[c_2 \ket{00} + c_3 \ket{01} + c_0\ket{10} + c_1\ket{11}\big] \\
		&\quad \xrightarrow{G_2} c_2\ket{00} + c_3 \ket{01} + c_0\ket{10} + c_1\ket{11} \big] \\
		&\quad \xrightarrow{G_3} -i\big[c_3 \ket{00} + c_2 \ket{01} + c_1\ket{10} + c_0\ket{11}\big] \\ 
		&\quad \xrightarrow{G_4} c_3 \ket{00} + c_2 \ket{01} + c_1\ket{10} + c_0\ket{11} \big].
	\end{align*}
	The total evolution time of this is $4\pi$.
	
	\begin{figure}
	\begin{center}
		\subfloat[] {
			\includegraphics{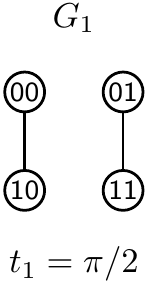} \quad
			\includegraphics{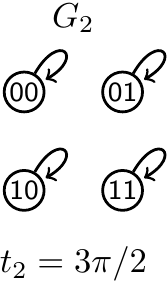} \quad
			\includegraphics{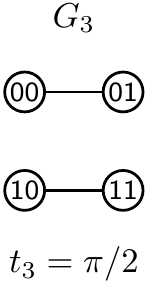} \quad
			\includegraphics{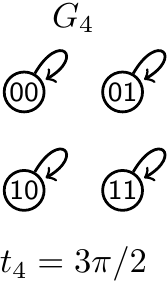}
			\label{fig:fullxx}
		}
	
		\subfloat[] {
			\includegraphics{XX_sequential_1} \quad
			\includegraphics{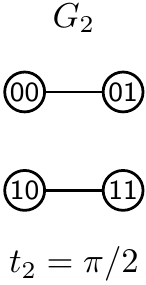} \quad
			\includegraphics{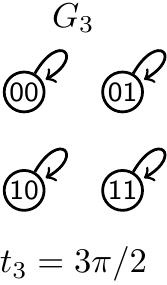} \quad
			\includegraphics{XX_sequential_4}
			\label{fig:commutingxx}
		}
	
		\subfloat[] {
			\includegraphics{XX_sequential_1} \quad
			\includegraphics{XX_comm_2} \quad
			\includegraphics{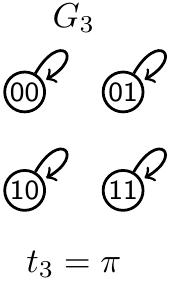}
			\label{fig:simplifiedxx}
		}
		
		\subfloat[] {
			\includegraphics{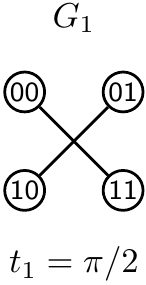} \quad
			\includegraphics{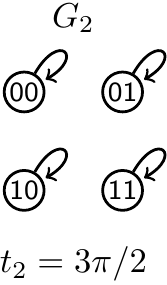} 
			\label{fig:simplifiedxxtwo}
		}
		\caption{Dynamic quantum walks for $X \otimes X$. (a) is the original, sequential implementation, while in (b), commuting graphs were swapped. In, (c) identical graphs were combined, and in (d) graphs with sequential perfect state transfers were combined.}
	\end{center}
	\end{figure}
	
	To begin simplifying this, we use Obs.~\ref{obs:commute}. Since $G_2$ and $G_3$ in \fref{fig:fullxx} commute, we can swap their order. This results in \fref{fig:commutingxx}. It acts on the initial state via
	\begin{align*}
		&c_0 \ket{00} + c_1 \ket{01} + c_2 \ket{10} + c_3 \ket{11}\\
		&\quad \xrightarrow{G_1} -i\big[c_2 \ket{00} + c_3 \ket{01} + c_0\ket{10} + c_1\ket{11}\big] \\
		&\quad \xrightarrow{G_2} -\big[c_3\ket{00} + c_2 \ket{01} + c_1\ket{10} + c_0\ket{11} \big] \\
		&\quad \xrightarrow{G_3} -i\big[c_3\ket{00} + c_2 \ket{01} + c_1\ket{10} + c_0\ket{11} \big] \\
		&\quad \xrightarrow{G_4} c_3 \ket{00} + c_2 \ket{01} + c_1\ket{10} + c_0\ket{11},
	\end{align*}
	which is the same final state as before.
	
	Although this alone does not result in a shorter dynamic graph, when it is used in conjunction with other observations, it can yield a dynamic quantum walk with a shorter evolution time and/ or fewer graphs, as we will see in the next two observations.
\end{example}


\subsection{Combining the Same Sequential Graphs}

The following observation describes how to combine two sequential, identical graphs.

\begin{observation}\label{obs:same}
	If $G_{\ell} = G_{\ell + 1}$, then $G_{\ell}$ can be performed for time $t_\ell+t_{\ell+1}$, modulo the period of graph $G_\ell$, and $G_{\ell + 1}$ can be omitted.
\end{observation}

\begin{proof}
	Since $G_{\ell} = G_{\ell + 1}$, $A_\ell = A_{\ell+1}$, and $\norm{A_\ell} = \norm{A_{\ell+1}}$. Then,
	\[ e^{-i A_{\ell+1} t_{\ell+1}/\norm{A_{\ell+1}}} e^{-i A_\ell t_\ell/\norm{A_\ell}} = e^{-i A_\ell (t_\ell+t_{\ell+1})/\norm{A_\ell}}, \]
	so we can replace both graphs with one graph and evolve by their combined evolution time, modulo the period of $G_\ell$.\qed
\end{proof}

\begin{example}\label{ex:xx}
	Let us continue the example of $X \otimes X$. After using Obs.~\ref{obs:commute}, we had the dynamic quantum walk pictured in \fref{fig:commutingxx}. Now, we can then use Obs.~\ref{obs:same} to combine $G_3$ and $G_4$ into a single graph with evolution time $3\pi/2 + 3\pi/2 = 3\pi \cong \pi \pmod{2\pi}$, since the period of $G_3$ is $2\pi$, to arrive at the walk in \fref{fig:simplifiedxx}. Then, the initial state evolves by
	\begin{align*}
		&c_0 \ket{00} + c_1 \ket{01} + c_2 \ket{10} + c_3 \ket{11}\\
		&\quad \xrightarrow{G_1} -i\big[c_2 \ket{00} + c_3 \ket{01} + c_0\ket{10} + c_1\ket{11}\big] \\
		&\quad \xrightarrow{G_2} -\big[c_3\ket{00} + c_2 \ket{01} + c_1\ket{10} + c_0\ket{11} \big] \\
		&\quad \xrightarrow{G_3} c_3 \ket{00} + c_2 \ket{01} + c_1\ket{10} + c_0\ket{11},
	\end{align*}
	which gives the same final state as before, but with one fewer graph and a total evolution time of $2\pi$, which is half the time of $4\pi$ of the previous implementations. With the next observation, we will be able to simplify the dynamic graph further still.
\end{example}


\subsection{Combining Sequential Perfect State Transfers}
	
The next observation can simplify a dynamic quantum walk when it contains a sequence of graphs where pairs of vertices have perfect state transfer.

\begin{observation}\label{obs:pst}
	Suppose the dynamic graph has $2^n$ vertices, so the vertices can be labeled in binary. Suppose the dynamic graph contains a sequence of graphs $\{ G_\ell, G_{\ell+1}, \dots, G_m \}$, such that for each graph in this sequence, there is perfect state transfer between every pair of vertices whose binary representations differ in the same locations. Then, the sequence of graphs simply performs a sequence of perfect state transfers, so we can replace $G_\ell$ through $G_m$ with a single graph that performs the resulting perfect state transfers. Any phases can be adjusted using isolated vertices with self-loops.
\end{observation}

\begin{proof}
	Suppose $G_\ell$ has perfect state transfer between vertices whose binary representations differ at digits $k_\ell$, where $k_\ell$ is a set of digits, as a pair of vertices can differ in more than one digit. Say vertices $\ket{a}$ and $\ket{b}$ is one pair of vertices where perfect state transfer occurs under $G_\ell$, and $\ket{c}$ and $\ket{d}$ is another pair. Then, if their initial amplitudes are $c_a$, $c_b$, $c_c$, and $c_d$, respectively, then after evolving on $G_\ell$, $\ket{a}$ has coefficient $c_b$, $\ket{b}$ has coefficient $c_a$, $\ket{c}$ has coefficient $c_d$, and $\ket{d}$ has coefficient $c_c$, all up to a phase.
	
	Now, suppose $G_{\ell+1}$ has perfect state transfer between vertices whose binary representations differ at digits $k_{\ell+1}$. Suppose $\ket{a}$ and $\ket{c}$ is a pair. Then after evolving on $G_{\ell+1}$, $\ket{a}$ has coefficient $c_d$, and $\ket{c}$ has coefficient $c_b$, all up to a phase. Since $a$ and $b$ differ in exactly digits $k_\ell$, and $a$ and $c$ differ in exactly digits $k_{\ell+1}$, $b$ and $c$ differ in digits $(k_\ell \cup k_{\ell+1}) \setminus (k_\ell \cap k_{\ell+1})$. The intersection is subtracted because flipping a digit and then flipping it again gives back the original value of digit. 

	Iterating this process, if there is a sequence of perfect state transfers between vertices that differ in digits $k_\ell, k_{\ell+1}, \dots k_m$, the index of the coefficient in front of state $\ket{a}$ is the bitstring that differs from $a$ in precisely the digits that appear in the union of the $k_i$'s an odd number of times. In fact, at the end of the sequence, each vertex in the graph has the initial coefficient of the vertex whose label in binary differs in exactly the positions of the digits that appear an odd number of times in the $k_i$'s, up to a phase. Thus, instead of evolving through all the perfect state transfers, we can just implement the resulting perfect state transfer by connect vertices that differ an odd number of times in the $k_i$'s by $P_2$'s for $t = \pi/2$. The phases can be corrected by evolving vertices as isolated vertices with self-loops.\qed
\end{proof}

\begin{example}
	Continuing the example of $X \otimes X$, after applying Obs.~\ref{obs:commute} and Obs.~\ref{obs:same}, we had \fref{fig:simplifiedxx}. We will simplify this further using our perfect state transfer observation.

	In $G_1$ of \fref{fig:simplifiedxx}, since the evolution time is $\pi/2$, we have perfect state transfer between vertices whose binary representations differ in their leftmost bits. That is, we have perfect state transfer between vertices 00 and 10, and between vertices 01 and 11, with an overall phase of $-i$. Then, in $G_2$ of \fref{fig:simplifiedxx}, we have perfect state transfer between vertices whose binary representations differ in their rightmost bits, so between vertices 00 and 01, and between vertices 10 and 11, again multiplying the phase of everything by $-i$, for a total phase of $(-i)^2 = -1$. From our observation, we can replace this sequence of perfect state transfers with single perfect state transfers. This sequence ultimately swaps the amplitudes at vertices 00 and 11, and the amplitudes at vertices 01 and 10. Thus, we can connect these pairs of vertices using $P_2$'s for time $\pi/2$ to achieve perfect state transfer between them, as shown in $G_1$ of \fref{fig:simplifiedxxtwo}. This creates an overall phase of $-i$ rather than $-1$, so to correct this, we evolve each vertex with a self-loop for time $\pi/2$ after evolving on $G_1$. We can combine this, however, with $G_3$ of \fref{fig:simplifiedxx} using Obs.~\ref{obs:same} since they are the same graph, resulting in $G_2$ of \fref{fig:simplifiedxxtwo}.

	To double-check our simplified dynamic graph, the walk on \fref{fig:simplifiedxxtwo} acts on the initial state via
	\begin{align*}
		&c_0 \ket{00} + c_1 \ket{01} + c_2 \ket{10} + c_3 \ket{11}\\
		&\quad \xrightarrow{G_1} -i\big[c_3 \ket{00} + c_2 \ket{01} + c_1\ket{10} + c_0\ket{11}\big] \\
		&\quad \xrightarrow{G_2} c_3 \ket{00} + c_2 \ket{01} + c_1\ket{10} + c_0\ket{11},
	\end{align*}
	which agrees with the previous implementations. Although the total runtime of $2\pi$ is the same as \fref{fig:simplifiedxx}, it uses one fewer graph. Altogether, we have simplified $X \otimes X$ from the original implementation in \fref{fig:fullxx}, where the $X$ gates were applied sequentially, to one where the gates are applied in parallel.
\end{example}


\subsection{Combining Sequential Complementary Subgraphs}

The next observation allows us to combine sequential graphs when the second graph is a subgraph of the complement of the first, and the second graph only has edges or self-loops on vertices that do not have edges or self-loops in the first graph.

\begin{observation}\label{obs:comp}
	Consider sequential graphs $G_\ell$ and $G_{\ell+1}$ with respective adjacency matrices $A_{\ell}$ and $A_{\ell + 1}$. Assume $G_{\ell+1}$ is a subgraph of the complement of $G_\ell$, and furthermore $G_{\ell+1}$ only contains edges or self-loops on vertices that do not have edges or self-loops in $G_\ell$. Then, if $\norm{A_\ell} = \norm{A_{\ell+1}}$, then $G_\ell$ and $G_{\ell+1}$ can be combined into one dynamic graph whose adjacency matrix is $A_\ell + A_{\ell+1}$, and if $t_\ell = t_{\ell+1}$, the evolution time of this graph is $t_\ell$. If $t_\ell \neq t_{\ell+1}$, the graphs can still be combined, but an additional graph will be needed to finish the propagation of the graph with the longer time. This results in a shorter overall time but the same number of graphs in the sequence.
\end{observation}

\begin{proof}
	Since $G_{\ell + 1}$ is contained in the complement of $G_{\ell}$ and only contains edges or self-loops on vertices that do not have edges or self-loops in $G_\ell$, the entries of $A_\ell + A_{\ell+1}$ are only 0's and 1's. Additionally, since $A_\ell$ and $A_{\ell+1}$ are adjacency matrices, they are symmetric, so their sum is also symmetric. So, $A_\ell + A_{\ell+1}$ is a valid adjacency matrix.
	
	Next, we prove that $A_\ell$ and $A_{\ell+1}$ commute. Let $V_\ell$ be the set of vertices in $G_\ell$ that have edges or self-loops, and let $V_\ell'$ be the set of vertices in $G_\ell$ that do not have edges or self-loops, so $V = V_\ell \cup V_\ell'$. Without loss of generality, we label the vertices in $V_\ell$ as $1, 2, \dots, |V_\ell|$ and the vertices of $V_\ell'$ as $|V_\ell|+1, \dots, |V_\ell|+|V_\ell'|$. Then the adjacency matrix $A_\ell$ has dimensions $(|V_1| + |V_2|) \times (|V_1| + |V_2|)$, and it has 1's and 0's in the upper left $|V_\ell| \times |V_\ell|$ corner and zeros everywhere else. In contrast, the adjacency matrix $A_{\ell+1}$ has 1's and 0's in the lower right $|V_{\ell+1}| \times |V_{\ell+1}|$ corner with 0's everywhere else. Then, $A_\ell A_{\ell+1} = A_{\ell+1} A_\ell = 0$, so they commute.

	Without loss of generality, say $t_\ell \ge t_{\ell+1}$. Then, if $\norm{A_\ell} = \norm{A_{\ell+1}}$,
	\begin{align*}
		& e^{-i A_{\ell+1} t_{\ell + 1} /\norm{A_{\ell+1}}} e^{-i A_\ell t_\ell /\norm{A_\ell}} \\ 
		&\quad= e^{-i A_\ell (t_\ell - t_{\ell+1}) /\norm{A_\ell}} e^{-i (A_\ell+A_{\ell+1}) t_{\ell+1}/\norm{A_\ell}}.
	\end{align*}
	Since $\norm{A_\ell + A_{\ell+1}} = \max \{\norm{A_\ell}, \norm{A_{\ell+1}}\} = \norm{A_\ell}$, where the first equality is from Section 1.3.7 of \cite{BH2012}, and the second equality is from $\norm{A_\ell} = \norm{A_{\ell+1}}$, the right exponential is a continuous-time quantum walk on the graph with adjacency matrix $A_\ell + A_{\ell+1}$. If $t_\ell \ne t_{\ell+1}$, then the left exponential is a quantum walk that finishes the propagation on the graph with the longer time, which we took to be $G_\ell$ without loss of generality.\qed
\end{proof}

\begin{example}
	We will illustrate this observation by acting on a qubit with $H$ and then acting on it with $Z$. This acts on the basis states as 
	\begin{align*}
		ZH \ket{0} &= \frac{1}{\sqrt{2}} (\ket{0} - \ket{1}), \\
		ZH \ket{1} &= \frac{1}{\sqrt{2}} (\ket{0} + \ket{1}).
	\end{align*}
	Implementing $H$ followed by $Z$ using the results from \cite{Wong33} gives the dynamic graph found in \fref{fig:fullzh}. It acts on the initial state as follows:
	\begin{align*}
		&c_0 \ket{0} + c_1 \ket{1}\\
		&\quad \xrightarrow{G_1} -ic_0 \ket{0} + c_1 \ket{1}\\
		&\quad \xrightarrow{G_2} \frac{1}{\sqrt{2}} \big[-i(c_0+c_1) \ket{0} + (-c_0+ c_1) \ket{1}\big] \\
		&\quad \xrightarrow{G_3} \frac{1}{\sqrt{2}} \big[i(c_0+c_1) \ket{0} + (c_0 - c_1) \ket{1}\big] \\ 
		&\quad \xrightarrow{G_4} \frac{1}{\sqrt{2}} \big[(c_0+c_1) \ket{0} + (c_0- c_1) \ket{1}\big]\\
		&\quad \xrightarrow{G_5} \frac{1}{\sqrt{2}} \big[(c_0+c_1) \ket{0} + (-c_0 + c_1) \ket{1}\big].
	\end{align*}
	The total evolution time is $13\pi/4 \approx 10.2$.
	
	This dynamic graph can be simplified by noting that $G_5$ is a subgraph of the complement of $G_4$, and furthermore that $G_5$ contains a self-loop on vertex 1, and vertex 1 does not have any edges or self-loops in $G_4$. Although $\norm{A_4} = \norm{A_5}$, their evolution times of $\pi/2$ and $\pi$ differ. So, we can only combine a part of $G_5$, leaving an additional graph to finish its propagation. 
	This is shown in \fref{fig:simplifiedzh}, where $G_4$ is the combined graph, and $G_5$ finishes the longer evolution. This final graph acts on the initial state via
	\begin{align*}
		&c_0 \ket{0} + c_1 \ket{1}\\
		&\quad \xrightarrow{G_1} -ic_0 \ket{0} + c_1 \ket{1}\\
		&\quad \xrightarrow{G_2} \frac{1}{\sqrt{2}} \big[-i(c_0+c_1) \ket{0} + (-c_0+ c_1) \ket{1}\big] \\
		&\quad \xrightarrow{G_3} \frac{1}{\sqrt{2}} \big[i(c_0+c_1) \ket{0} + (c_0 - c_1) \ket{1}\big] \\ 
		&\quad \xrightarrow{G_4} \frac{1}{\sqrt{2}} \big[(c_0+c_1) \ket{0} -i (c_0- c_1) \ket{1}\big]\\
		&\quad \xrightarrow{G_5} \frac{1}{\sqrt{2}} \big[(c_0+c_1) \ket{0} + (-c_0 + c_1) \ket{1}\big].
	\end{align*}
	The two walks give the same final state. The total evolution time of this simplified graph is $11\pi/4 \approx 8.6$, which is a roughly 15\% speedup over the original implementation.
	
	\begin{figure}
	\begin{center}
		\subfloat[] {
			\includegraphics{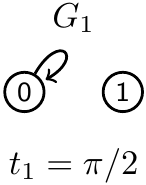} \quad
			\includegraphics{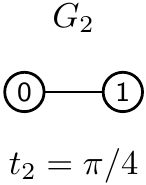} \quad
			\includegraphics{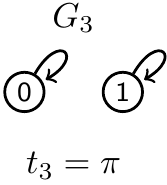} \quad
			\includegraphics{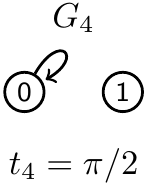} \quad
			\includegraphics{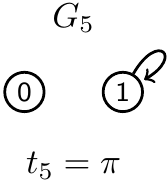}
			\label{fig:fullzh}
		}
	
		\subfloat[] {
			\includegraphics{ZH_sequential_1} \quad
			\includegraphics{ZH_sequential_2} \quad
			\includegraphics{ZH_sequential_3} \quad
			\includegraphics{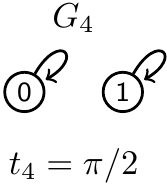} \quad
			\includegraphics{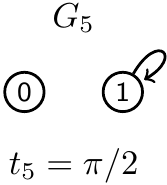}
			\label{fig:simplifiedzh}
		}
		\caption{(a) The original dynamic quantum walk for $H$ acting on one qubit followed by $Z$ acting on one qubit. (b) The simplified dynamic graph using the complementary graph observation.}
	\end{center}
	\end{figure}
\end{example}


\subsection{Moving Looped Singletons}

The fifth observation allows us to change the order in which some isolated vertices with self-loops, which we call \emph{looped singletons} for brevity, are allowed to accumulate phases.
\begin{observation}\label{obs:singletons}
	If a vertex $v$ is propagated as a looped singleton in $G_\ell$ for time $t_\ell$, and vertex $v$ is not adjacent to any other vertices in $G_j$ for all $a \leq j \leq b$, we can instead propagate vertex $v$ as a looped singleton in any $G_j$ for time $t_j$ between $G_a$ and $G_b$ that satisfies $t_j/\norm{A_j} = t_\ell/\norm{A_\ell}$. If we choose to propagate it as a looped singleton in a graph $G_j$ where it is already connected as a looped singleton, we evolve the singleton by time $2t_j$, modulo the period of $G_j$.
\end{observation}

\begin{proof}
	Without loss of generality, suppose vertex $v = \ket{0 \dots 0}$ is a looped singleton in $G_\ell$. Let $A_\ell$ be the adjacency matrix of the graph $G_\ell$. Since $v$ is a looped singleton, we can rewrite this matrix as the sum of two matrices, $A_v$ and $A_\ell'$, where $A_v$ is the matrix with a 1 in its upper-left corner and zeros everywhere else. Let $A_\ell' = A_\ell - A_v$. Since $v$ is a looped singleton, the topmost row and leftmost column of $A_\ell'$ consists of all 0's. Then, $A_v A'_\ell = A'_\ell A_v = 0$, so $A_v$ and $A'_\ell$ commute. Then,
	\begin{align*}
		e^{-i A_\ell t_\ell / \norm{A_\ell}}
			&= e^{-i (A_v + A_\ell')t_\ell / \norm{A_\ell}} \\
			&= e^{-i A_v t_\ell / \norm{A_\ell}} e^{-i A_\ell' t_\ell / \norm{A_\ell}}.
	\end{align*}
	So, we have split $G_\ell$ into two graphs, one without the looped singleton and with adjacency matrix $A'_\ell$ and evolution time $t_\ell$, and another with the looped singleton and with adjacency matrix $A_v$, which evolves for time $t_\ell / \norm{A_\ell}$ (since $\norm{A_v} = 1$).

	Now, say vertex $v$ does not share an edge with any other vertex in all $G_j$, where $m \ge j > \ell$. Then, the topmost row and leftmost column of $A_j$ consists of all 0's, except possibly the top-left corner, which is 0 if vertex $v$ does not have a self-loop in $G_j$, and 1 if it does have a self-loop in $G_j$. Then, $A_v A_j = A_j A_v = 0$ if the top-left corner of $A_j$ is 0, and $A_v A_j = A_j A_v = A_v$ if the top-left corner of $A_j$ is 1. Either way, $A_v$ and $A_j$ commute. Thus, $e^{-i A_v t_\ell / \norm{A_\ell}}$ commutes with all $e^{-i A_j t_j / \norm{A_j}}$, so we can move the looped singleton through all the $G_j$'s.
	
	Say we moved the looped singleton next to one of the $G_j$'s, so the relevant part of the time-evolution operator is
	\[ e^{-iA_jt_j/\norm{A_j}} e^{-iA_vt_\ell/\norm{A_\ell}} = e^{-i \left( A_j t_j + A_v t_\ell \norm{A_j}/\norm{A_\ell} \right) / \norm{A_j}}. \]
	Now, we consider two cases. First, $A_j$ is also just a looped singleton at $v$ with no other edges, then $A_j = A_v$, and $\norm{A_j} = \norm{A_v} = 1$, so the previous exponential becomes
	\[ e^{-i A_v \left( t_j + t_\ell/\norm{t_\ell} \right)}. \]
	Thus, the two looped singleton graphs can combined into a single one with evolution time $t_j + t_\ell/\norm{A_\ell}$, modulo the period of $G_j$.
	
	Second, if $v$ is not a looped singleton in $G_j$, and furthermore if $t_j/\norm{A_j} = t_\ell/\norm{A_\ell}$, then the time-evolution of $G_j$ and the looped singleton becomes
	\[ e^{-i \left( A_j + A_v \right) t_j / \norm{A_j}}. \]
	Note $A_j + A_v$ is a valid adjacency matrix, and it describes graph $G_j$ combined with the looped singleton, and it evolves for time $t_j$. Note adding the looped singleton does not affect $\norm{A_j}$.
	
	Finally, note if $t_j/\norm{A_j} \ne t_\ell/\norm{A_\ell}$ because $t_\ell$ is too long, we can split the evolution time of the looped singleton up into $t_\ell = t_1 + t_2$, where $t_j/\norm{A_j} = t_1/\norm{A_\ell}$. Then, we can combine that with $A_j$, and we have another graph that finishes the phase accumulation of the looped singleton for time $t_2$.\qed
\end{proof}

\begin{example}\label{ex:yz}
	For this example, we consider $Y \otimes Z$. This acts on the computational basis states as 
	\begin{align*}
		(Y \otimes Z) \ket{00} &= -i \ket{10}, \\
		(Y \otimes Z) \ket{01} &= i\ket{11}, \\
		(Y \otimes Z) \ket{10} &= i\ket{00} , \\
		(Y \otimes Z) \ket{11} &= -i \ket{01}.
	\end{align*}
	Implementing this sequentially using $Y \otimes I$ followed by $I \otimes Z$ using the results from \cite{Wong33}, we get the dynamic graph shown in \fref{fig:fullyz}. It evolves the initial state according to
	\begin{align*}
		&c_0 \ket{00} + c_1 \ket{01} + c_2 \ket{10} + c_3 \ket{11}\\
		&\quad \xrightarrow{G_1} -i\big[c_2 \ket{00} + c_3 \ket{01} + c_0\ket{10} + c_1\ket{11}\big] \\
		&\quad \xrightarrow{G_2} i\big[-c_2\ket{00} - c_3 \ket{01} + c_0\ket{10} + c_1\ket{11} \big] \\
		&\quad \xrightarrow{G_3} i\big[-c_3 \ket{00} + c_2 \ket{01} + c_1\ket{10} - c_0\ket{11}\big].
	\end{align*}
	The total evolution time of this is $5\pi/2$.
	
	\begin{figure}
	\begin{center}
		\subfloat[]{
			\includegraphics{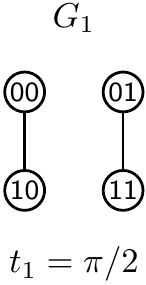}\quad
			\includegraphics{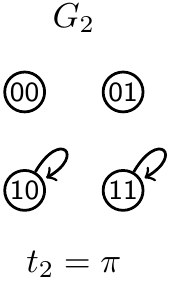}\quad
			\includegraphics{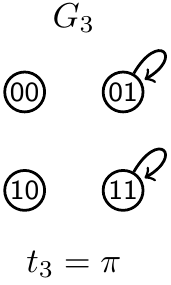}
			\label{fig:fullyz}
		}
		
		\subfloat[]{
			\includegraphics{YZ_1}\quad
			\includegraphics{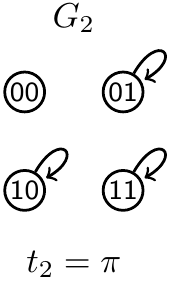}\quad
			\includegraphics{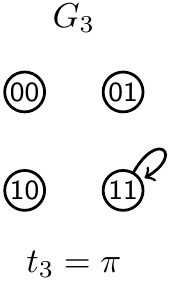}
			\label{fig:intermediateyz}
		}
	
		\subfloat[]{
			\includegraphics{YZ_1}\quad
			\includegraphics{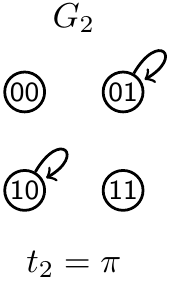}\quad
			\includegraphics{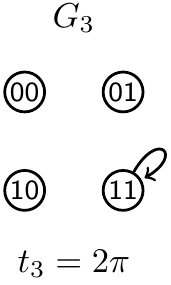}
			\label{fig:interyz}
		}
	
		\subfloat[]{
			\includegraphics{YZ_1} \quad
			\includegraphics{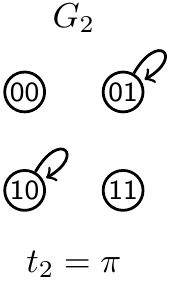}
			\label{fig:simplifiedyz}
		}

		\caption{(a) The dynamic quantum walk for $Y \otimes I$ followed by $I \otimes Z$. (b) The dynamic graph after moving a singleton. (c) The dynamic graph after moving another singleton. (d) The dynamic graph after eliminating a graph.}
	\end{center}
	\end{figure}
	
	We can simplify this dynamic graph by moving the looped singletons. First, we can move the self-loop at $\ket{01}$ in $G_3$ to $G_2$. This is because $\ket{01}$ does not have a self-loop in $G_2$, and $t_3/\norm{A_3} = t_2/\norm{A_3}$. This results is \fref{fig:intermediateyz}. Next, we move the self-loop at $\ket{11}$ in $G_2$ to $G_3$. Since $G_3$ is just a looped singleton at $\ket{11}$, we can combine these with a total evolution time of $\pi + \pi/1 = 2 \pi$. This yields \fref{fig:interyz}. Finally, the evolution time of $G_3$ can be taken modulo the period of $G_3$, which is $2\pi$, so the evolution time is zero. So, we can drop $G_3$ entirely, resulting in \fref{fig:simplifiedyz}. As a check, it evolves the initial state according to
	\begin{align*}
		&c_0 \ket{00} + c_1 \ket{01} + c_2 \ket{10} + c_3 \ket{11}\\
		&\quad \xrightarrow{G_1} -i\big[c_2 \ket{00} + c_3 \ket{01} + c_0\ket{10} + c_1\ket{11}\big] \\
		&\quad \xrightarrow{G_2}  i\big[-c_3\ket{00} + c_2 \ket{01} + c_1\ket{10} - c_0\ket{11} \big].
	\end{align*}
	This gives the same end result as $Y \otimes I$ followed by $I \otimes Z$, but with one fewer graph and an overall time of $3\pi/2$, which is a speedup of 40\%.
\end{example}


\subsection{Implementing Hadamards using Hypercube Uniform Mixing}

The final observation about mixing on hypercubes allows us to efficiently implement the Hadamard gate acting on $n$ qubits, which appears in several basic quantum algorithms such as Deutsch's algorithm for parity, the Deutsch-Jozsa algorithm for distinguishing between constant and balanced functions, and the Bernstein-Vazirani algorithm for determining a binary string that an oracle function dot-products with its input \cite{NielsenChuang2000}.
\begin{observation}\label{obs:hypercube}
	We can use uniform mixing on the hypercube to apply the Hadamard gate to multiple qubits in parallel.
\end{observation}

\begin{proof}
	We begin with how $H^{\otimes n}$ acts on the state $c_0 \ket{0} + \dots + c_{N-1} \ket{N-1}$, where $N = 2^n$. First, $H^{\otimes n}$ is an $N \times N$ matrix, and its elements are
	\[ (H^{\otimes n})_{i,j} = \frac{(-1)^{i \cdot j}}{2^{n/2}}, \]
	where we label each row and column from $0$ to $N-1$, so the top left entry of the matrix corresponds to $(H^{\otimes n})_{0,0}$. Furthermore, $i \cdot j = \sum_m i_m j_m$ denotes the bitwise dot product of $i$ and $j$, where $i_m$ and $j_m$ are the $m^{th}$ digits in the binary representation of $i$ and $j$, respectively. The Hadamard gates act on the initial state by
	\begin{align}
		&H^{\otimes n} \left( c_0 \ket{0} + \dots + c_{N-1} \ket{N-1} \right) \nonumber \\
		&\quad= \frac{c_0}{2^{n/2}} \left[ (-1)^{0 \cdot 0} \ket{0} + \dots + (-1)^{(N-1) \cdot 0} \ket{N-1} \right] \nonumber \\
		&\quad\quad+ \frac{c_1}{2^{n/2}} \left[ (-1)^{0 \cdot 1} \ket{0} + \dots + (-1)^{(N-1) \cdot 1} \ket{N-1} \right] + \dots \nonumber \\
		&\quad\quad+ \frac{c_{N-1}}{2^{n/2}} \Big[ (-1)^{0 \cdot (N-1)} \ket{0} + \dots \nonumber \\
		&\quad\quad\quad\quad\quad\quad+ (-1)^{(N-1) \cdot (N-1)} \ket{N-1} \Big] \nonumber \\
		&\quad= \frac{1}{2^{n/2}} \bigg\{ \left[ c_0 (-1)^{0 \cdot 0} + \dots + c_{N-1} (-1)^{0 \cdot (N-1)} \right] \ket{0} \label{eq:hypercube_hadamards} \\
		&\quad\quad+ \left[ c_0 (-1)^{(1) \cdot 0} + \dots + c_{N-1} (-1)^{1 \cdot (N-1)} \right] \ket{1} + \dots \nonumber \\
		&\quad\quad+ \Big[ c_0 (-1)^{(N-1) \cdot 0} + \dots \nonumber \\
		&\quad\quad\quad\quad+ c_{N-1} (-1)^{(N-1) \cdot (N-1)} \Big] \ket{N-1} \bigg\}. \nonumber
	\end{align}

	Next, let us explore how to implement this using a quantum walk on the $n$-dimensional hypercube. It begins with seminal work by Moore and Russell, who explored uniform mixing on the $n$-dimensional hypercube with Hamiltonian $H = A/n$ \cite{moore2002quantum}. Since $\norm{A} = n$ for the hypercube, this is $H = A/\norm{A}$, which is the same Hamiltonian that we are using. In their Appendix C, they showed that if the walker is initially at vertex $\ket{00\dots0}$, i.e., $\ket{\psi(0)} = \ket{0}^{\otimes n}$, then after walking on the hypercube for time $t$, the state of the system is
	\[ \ket{\psi(t)} = \left[ \cos \left( \frac{t}{n} \right) \ket{0} + i \sin \left( \frac{t}{n} \right) \ket{1} \right]^{\otimes n}. \]
	Continuing their Appendix C, since a vertex with Hamming weight $x$ has $(n-x)$ zeros and $x$ ones, the amplitude at vertex $x$ is
	\[ \braket{x}{\psi(t)} = i^x \sin^x \left( \frac{t}{n} \right) \cos^{n-x} \left( \frac{t}{n} \right). \]
	Moore and Russell noted when $t = kn\pi/4$ for odd $k$, we get the uniform distribution, i.e., the probability at each vertex $x$ is $|\braket{x}{\psi(t)}|^2 = 1/2^n$ for all $x$.
	
	For our purposes, we want to implement gates in the shortest time possible, so we take $k = 1$. That is, at time $t = n\pi/4$, the amplitude at vertex $x$ is
	\[ \braket{x}{\psi(t)} = \frac{i^x}{2^{n/2}}. \]
	As an extension of Moore and Russel's result, we can change the starting state so that instead of starting at vertex $\ket{0}^{\otimes n}$, we start at some vertex $\ket{v}$. Then, at time $t = n\pi/4$, the phase of the amplitude at each vertex depends on its Hamming distance from vertex $\ket{v}$. That is, the amplitude at vertex $x$ is
	\[ \braket{x}{\psi(t)} = \frac{i^{d_H(x,v)}}{2^{n/2}}, \]
	where $d_H(x,v)$ denotes the Hamming distance between vertices $x$ and $v$. Generalizing a step further, say the initial state is a superposition over the vertices, $\ket{\psi(0)} = c_0 \ket{0} + \dots + c_{N-1}\ket{N-1}$. Then, at time $t = n\pi/4$, a fraction of $c_0$ (specifically $1/2^{n/2}$) ends up at all the other vertices in the hypercube, with a phase depending on the vertex's Hamming distance from vertex $0 \pmod 4$. Similarly, from vertex 1, a fraction of $c_1$ would jump to all the other vertices in the hypercube with a phase depending on the Hamming distance from vertex $1 \pmod 4$. Adding the contribution from each vertex, the quantum walk on the hypercube transforms the initial state by
	\begin{align*}
		&e^{-iA\pi/4} \left( c_0 \ket{0} + \dots + c_{N-1} \ket{N-1} \right) \\
		&\quad= \frac{c_0}{2^{n/2}} \left[ i^{d_H(0,0)} \ket{0} + \dots + i^{d_H(N-1,0)} \ket{N-1} \right] \\
		&\quad\quad+ \frac{c_1}{2^{n/2}} \left[ i^{d_H(0,1)} \ket{0} + \dots + i^{d_H(N-1,1)} \ket{N-1} \right] + \dots \\
		&\quad\quad+ \frac{c_{N-1}}{2^{n/2}} \left[ i^{d_H(0,N-1)} \ket{0} + \dots + i^{d_H(N-1,N-1)} \ket{N-1} \right] \\
		&\quad= \frac{1}{2^{n/2}} \bigg\{ \left[ c_0 i^{d_H(0,0)} + \dots + c_{N-1} i^{d_H(0,N-1)} \right] \ket{0} \\
		&\quad\quad+ \left[ c_0 i^{d_H(1,0)} + \dots + c_{N-1} i^{d_H(1,N-1)} \right] \ket{1} + \dots \\
		&\quad\quad+ \left[ c_0 i^{d_H(N-1,0)} + \dots + c_{N-1} i^{d_H(N-1,N-1)} \right] \ket{N-1} \bigg\}.
	\end{align*} 
	This does mix the amplitudes, but it does not quite match the result of the Hadamard matrices in \eqref{eq:hypercube_hadamards}.
	
	To fix the discrepancy, we can adjust the phase of each vertex by evolving by looped singletons before and after walking on the hypercube. Say we start with the initial state $c_0 \ket{0} + \dots + c_{N-1} \ket{N-1}$. If we evolve by looped singletons at each vertex for various amounts of time, we can apply phases $\omega_i$ to each vertex $i$. Then, the state of the quantum walk is
	\[ c_0 \omega_0 \ket{0} + c_1 \omega_1 \ket{1} + \dots + c_{N-1} \omega_{N-1} \ket{N-1}. \]
	Then, if we walk on the hypercube for time $t = n\pi/4$, we get
	\begin{align*}
		&\frac{1}{2^{n/2}} \bigg\{ \left[ c_0 \omega_0 i^{d_H(0,0)} + \dots + c_{N-1} \omega_{N-1} i^{d_H(0,N-1)} \right] \ket{0} \\
		&\quad+ \left[ c_0 \omega_0 i^{d_H(1,0)} + \dots + c_{N-1} \omega_{N-1} i^{d_H(1,N-1)} \right] \ket{1} + \\
		&\quad\dots \\
		&\quad+ \Big[ c_0 \omega_0 i^{d_H(N-1,0)} + \dots \\
		&\quad\quad\quad + c_{N-1} \omega_{N-1} i^{d_H(N-1,N-1)} \Big] \ket{N-1} \bigg\}.
	\end{align*}
	Now, we can evolve by looped singletons again, applying phases $\omega'_i$ to each vertex, resulting in
	\begin{align*}
		&\frac{1}{2^{n/2}} \bigg\{ \omega'_0 \left[ c_0 \omega_0 i^{d_H(0,0)} + \dots + c_{N-1} \omega_{N-1} i^{d_H(0,N-1)} \right] \ket{0} \\
		&\quad+ \omega'_1 \left[ c_0 \omega_0 i^{d_H(1,0)} + \dots + c_{N-1} \omega_{N-1} i^{d_H(1,N-1)} \right] \ket{1} \\
		&\quad+ \dots \\
		&\quad+ \omega'_{N-1} \Big[ c_0 \omega_0 i^{d_H(N-1,0)} + \dots \\
		&\quad\quad\quad\quad\quad + c_{N-1} \omega_{N-1} i^{d_H(N-1,N-1)} \Big] \ket{N-1} \bigg\}.
	\end{align*}
	Comparing this to \eqref{eq:hypercube_hadamards}, the final states of the dynamic quantum walk and the Hadamard gates are the same when
	\begin{equation}
		\label{eq:hypercube_condition}
		\omega'_x \omega_y i^{d_H(x,y)} = (-1)^{x \cdot y}
	\end{equation}
	for all $x,y = 0, 1, \dots, N-1$. To make these equal, we apply the following phase to each vertex $v$ both before and after walking on the hypercube:
	\[ \omega_v = \omega'_v = \begin{cases}
		-1, & h(v) \cong 0 \pmod 4 \\
		-i, & h(v) \cong 1 \pmod 4 \\
		1, & h(v) \cong 2 \pmod 4 \\
		i, & h(v) \cong 3 \pmod 4 \\
	\end{cases} = -e^{-i h(v)\pi/2}, \]
	where $h(v)$ is the Hamming weight of $v$. Plugging into the left hand side of \eqref{eq:hypercube_condition} and also using $i = e^{i\pi/2}$, we get
	\begin{align*}
		\omega'_x \omega_y i^{d_H(x,y)}
			&= e^{-i h(x)\pi/2} e^{-i h(y)\pi/2} e^{i d_H(x,y)\pi/2} \\
			&= e^{i(d_H(x,y) - h(x) - h(y)) \pi/2}.
	\end{align*}
	The exponent can be related to the dot product of $x$ and $y$, which is the number of 1's they have in common. To derive this, recall $h(x)$ is the number of 1's in bitstring $x$, and $h(y)$ is the number of 1's in bitstring $y$. Since $x \cdot y$ is the number of pairs of 1's they have in the same position, bitstring $x$ has $[h(x) - x \cdot y]$ 1's that are not paired with a 1 in $y$, so it is paired with a 0. Similarly, bitstring $y$ has $[h(y) - x \cdot y]$ 1's that are not paired with a 1 in $x$, so it is paired with a 0. Adding these, we get the total number of 1's paired with a 0, which is the number of positions where the bits differ, which is the Hamming distance $d_H(x,y)$:
	\[ [h(x) - x \cdot y] + [h(y) - x \cdot y] = d_H(x,y). \]
	For example, say we have two bytes $x$ and $y$. Say $x$ has three 1's, $y$ has six 1's, and there are two pairs of 1's in the same position, such as
	\begin{align*}
		x &= 11100000 \\
		y &= 11001111.
	\end{align*} 
	Then $x$ has $3-2 = 1$ remaining 1 that must be paired with a 0 in $y$, and $y$ has $6 - 2 = 4$ remaining 1's that must be paired with a 0 in $x$. So, their Hamming distance is is $1 + 4 = 5$. Returning to the general expression, we can rearrange it to get
	\[ d_H(x,y) - h(x) - h(y) = -2 x \cdot y. \]
	Thus,
	\[ e^{i(d_H(x,y) - h(x) - h(y)) \pi/2} = e^{-i x \cdot y \pi} = (-1)^{x \cdot y}, \]
	which is the right-hand side of \eqref{eq:hypercube_condition}. Thus, with phase adjustments of $e^{-ih(v)\pi/2}$ before and after walking on the hypercube, which can be implemented by evolving vertex $v$ as a looped singleton for time $h(v)\pi/2 \pmod{2\pi}$, since the period of a graph with only isolated vertices and singletons is $2\pi$ we can implement parallel Hadamard gates.
	
	Generalizing, we can apply Hadamard gates to a subset of the qubits rather than all the qubits by walking on a smaller-dimensional hypercubes, again with phase adjustmest before and after.\qed
\end{proof}

\begin{example}\label{ex:hh}
	We will look at $H \otimes H$ in this example. It maps
	\begin{align*}
		(H \otimes H) \ket{00} &= \frac{1}{2} \left( \ket{00} + \ket{01} + \ket{10} + \ket{11} \right), \\
		(H \otimes H) \ket{01} &= \frac{1}{2} \left( \ket{00} - \ket{01} + \ket{10} - \ket{11} \right), \\
		(H \otimes H) \ket{10} &= \frac{1}{2} \left( \ket{00} + \ket{01} - \ket{10} - \ket{11} \right), \\
		(H \otimes H) \ket{11} &= \frac{1}{2} \left( \ket{00} - \ket{01} - \ket{10} + \ket{11} \right).
	\end{align*}
	The dynamic quantum walk for $H \otimes I$ followed by $I \otimes H$, using results from \cite{Wong33}, is shown in \fref{fig:fullhh}. The effect of each graph on the initial state is
	\begin{align*}
		&c_0 \ket{00} + c_1 \ket{01} + c_2 \ket{10} + c_3 \ket{11}\\
		&\quad \xrightarrow{G_1} c_0 \ket{00} + c_1 \ket{01} + ic_2\ket{10} + ic_3\ket{11} \\
		&\quad \xrightarrow{G_2} \frac{1}{\sqrt{2}} \big[ (c_0+c_2) \ket{00} + (c_1+c_3) \ket{01} + i(-c_0+c_2) \ket{10} + i(-c_1+c_3) \ket{11} \big] \\
		&\quad \xrightarrow{G_3} \frac{1}{\sqrt{2}} \big[ (c_0+c_2) \ket{00} + (c_1+c_3) \ket{01} + (c_0-c_2) \ket{10} + (c_1-c_3) \ket{11} \big] \\
		&\quad \xrightarrow{G_4} \frac{1}{\sqrt{2}} \big[ (c_0+c_2) \ket{00} + i(c_1+c_3) \ket{01} + (c_0-c_2) \ket{10} + i(c_1-c_3) \ket{11} \big] \\
		&\quad \xrightarrow{G_5} \frac{1}{2} \big[ (c_0+c_1+c_2+c_3) \ket{00} + i(-c_0+c_1-c_2+c_3) \ket{01} \\
		&\quad\quad\quad+ (c_0+c_1-c_2-c_3) \ket{10} + i(-c_0+c_1+c_2-c_3) \ket{11} \big] \\
		&\quad \xrightarrow{G_6} \frac{1}{2} \big[ (c_0+c_1+c_2+c_3) \ket{00} + (c_0-c_1+c_2-c_3) \ket{01} \\
		&\quad\quad\quad+ (c_0+c_1-c_2-c_3) \ket{10} + (c_0-c_1-c_2+c_3) \ket{11} \big].
	\end{align*}
	This has a total evolution time of $13\pi/2 \approx 20.4$.

	\begin{figure}
	\begin{center}
		\subfloat[] {
			\includegraphics{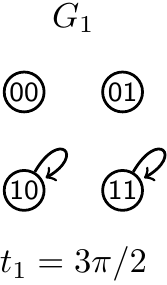} \quad
			\includegraphics{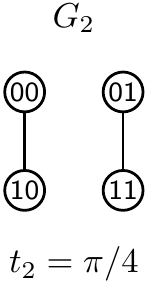} \quad
			\includegraphics{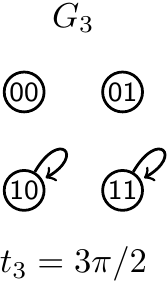} \quad
			\includegraphics{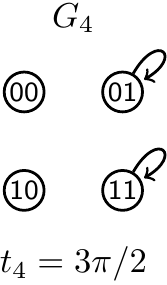} \quad
			\includegraphics{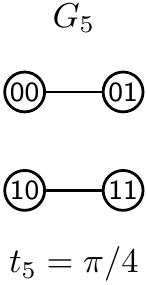} \quad
			\includegraphics{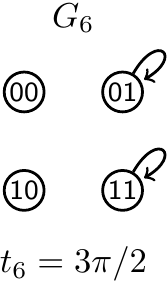}
			\label{fig:fullhh}
		}

		\subfloat[] {
			\includegraphics{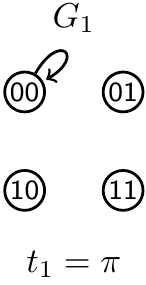} \quad
			\includegraphics{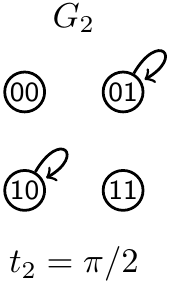} \quad
			\includegraphics{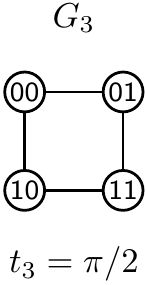} \quad
			\includegraphics{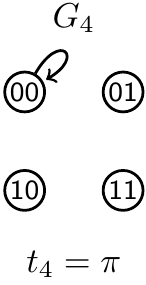} \quad
			\includegraphics{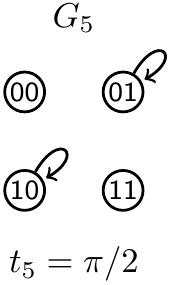}
			\label{fig:simplifiedhh}
		}
		
		\subfloat[] {
			\includegraphics{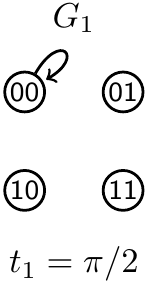} \quad
			\includegraphics{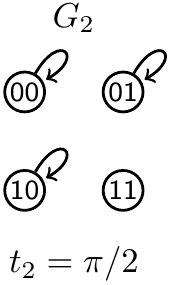} \quad
			\includegraphics{HH_hypercube_3} \quad
			\includegraphics{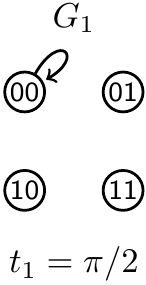} \quad
			\includegraphics{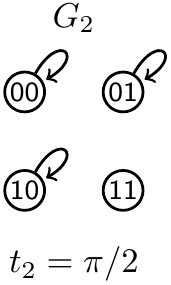}
			\label{fig:finalhh}
		}
		
		\caption{(a) The dynamic quantum walk for $H \otimes I$ followed by $I \otimes H$. (b) The dynamic quantum walk for $H \otimes H$ using the hypercube observation. (c) The dynamic quantum walk further simplified by moving looped singletons.}
	\end{center}
	\end{figure}

	Using the hypercube observation, we can implement $H \otimes H$ using the dynamic graph found in \fref{fig:simplifiedhh}. In $G_1$, we evolve vertex $\ket{00}$ as a looped singleton for time $\pi$ so that it acquires a phase of $-1$. In $G_2$, we evolve vertices $\ket{01}$ and $\ket{10}$ for time $\pi/2$ so that they each acquire a phase of $-i$. In $G_3$, we walk on the 2D hypercube, which is a square. In $G_4$ and $G_5$, we repeat the phases in $G_1$ and $G_2$. To see how the hypercube implementation affects the initial state, 
	\begin{align*}
		&c_0 \ket{00} + c_1 \ket{01} + c_2 \ket{10} + c_3 \ket{11}\\
		&\quad \xrightarrow{G_1} -c_0 \ket{00} + c_1 \ket{01} + c_2\ket{10} + c_3\ket{11} \\
		&\quad \xrightarrow{G_2} -c_0\ket{00} - ic_1 \ket{01} - ic_2\ket{10} + c_3\ket{11}\\
		&\quad \xrightarrow{G_3} \frac{1}{2} \big[ -(c_0+c_1+c_2+c_3) \ket{00}+ i(c_0-c_1+c_2-c_3) \ket{01} \\ 
		&\quad\quad\quad+ i(c_0+c_1-c_2-c_3) \ket{10} + (c_0-c_1-c_2+c_3) \ket{11} \big] \\
		&\quad \xrightarrow{G_4} \frac{1}{2} \big[ (c_0+c_1+c_2+c_3) \ket{00}+ i(c_0-c_1+c_2-c_3) \ket{01} \\ 
		&\quad\quad\quad+ i(c_0+c_1-c_2-c_3) \ket{10} + (c_0-c_1-c_2+c_3) \ket{11} \big] \\
		&\quad \xrightarrow{G_5} \frac{1}{2} \big[ (c_0+c_1+c_2+c_3) \ket{00} + (c_0-c_1+c_2-c_3) \ket{01} \\ 
		&\quad\quad\quad+ (c_0+c_1-c_2-c_3) \ket{10} + (c_0-c_1-c_2+c_3) \ket{11} \big].
	\end{align*}
	This is the same final result as implementing $H \otimes I$ followed by $I \otimes H$. The total time for this implementation is $7\pi/2 \approx 11.0$, and it consists of five graphs. This is 46\% faster than the $26\pi/4 \approx 20.4$ needed to implement $H \otimes I$ followed by $I \otimes H$, and it uses one fewer graph.
	
	Note we can further use Obs.~\ref{obs:singletons} to partially combine $G_1$ and $G_2$, as well as $G_4$ and $G_5$, resulting in \fref{fig:finalhh}. This has a total evolution time of $5\pi/2 \approx 7.9$, which is a 62\% speedup over the original, sequential implementation.
\end{example}

Note the speedup provided by this observation increases as we increase the number of Hadamard gates. Using the sequential implementation from previous works \cite{Wong33}, a Hadamard gate takes 3 graphs and $13\pi/4$ time. Then, $H^{\otimes n}$ takes $3n$ graphs and $13n\pi/4$ time. Using our hypercube approach, $H^{\otimes n}$ only takes 5 graphs. That is, it takes 2 graphs to apply phases, one for time $\pi$ and another for time $\pi/2$. Together, these allow us to apply all possible phases that we need. Then, we have one graph where we walk on the hypercube for time $n\pi/4$.  Then we apply phases again with the two graphs. So the total time is $3+n\pi/4$. For large $n$, this is $n\pi/4$, which is a 92.3\% speedup over the sequential implementation's $13n\pi/4$.


\section{Example: Simplification of the dynamic quantum walk for \fref{fig:sequential}}\label{sec:simplification}

In this section, we return to the motivating example from the introduction, which was whether there is a simpler way to implement the circuit shown in \fref{fig:sequential} using a quantum walk on a dynamic graph. To do this, we begin with the sequential implementation and then simplify it using the observations introduced in this paper. The sequential implementation comes from \cite{Wong33}, and it is shown in \fref{fig:fullcircuit}. $G_1$ through $G_3$ apply $H \otimes I \otimes I$, then $G_4$ and $G_5$ apply $I \otimes X \otimes I$, then $G_6$ through $G_8$ apply $I \otimes I \otimes H$, etc. It contains sixteen graphs, and the total evolution time is $67\pi/4 \approx 52.6$.

\begin{figure*}
\begin{center}
	\subfloat[] {
		\label{fig:fullcircuit}
		\begin{minipage}[b]{1.0\linewidth}\centering
			\includegraphics[scale=.75]{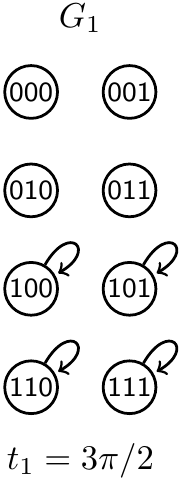} \thinspace
			\includegraphics[scale=.75]{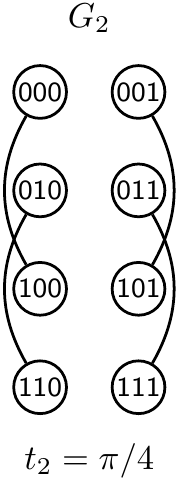} \thinspace
			\includegraphics[scale=.75]{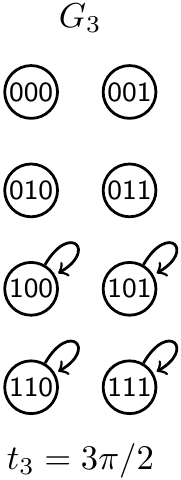} \thinspace
			\includegraphics[scale=.75]{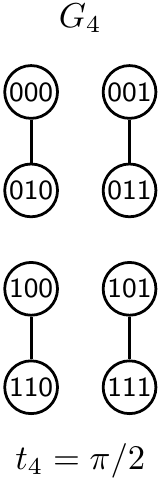} \thinspace
			\includegraphics[scale=.75]{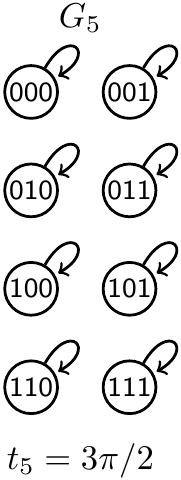} \thinspace
			\includegraphics[scale=.75]{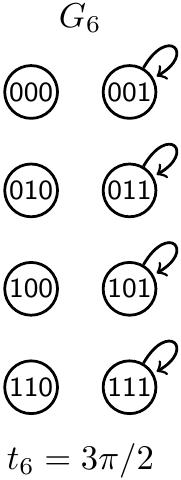} \thinspace
			\includegraphics[scale=.75]{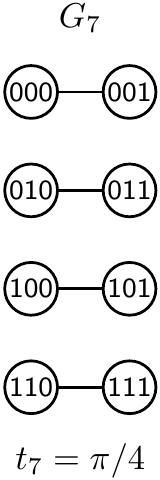} \thinspace
			\includegraphics[scale=.75]{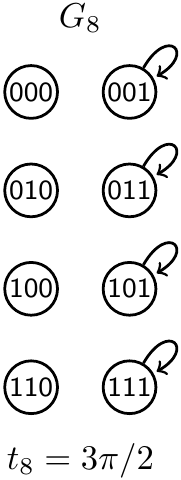} \hfill
			
			\vspace{1em}
			\includegraphics[scale=.75]{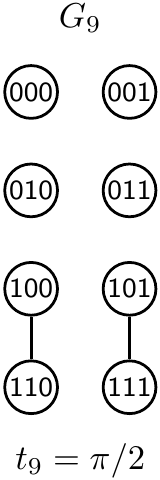} \thinspace
			\includegraphics[scale=.75]{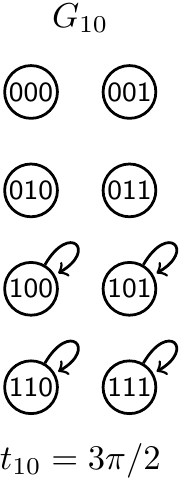} \thinspace
			\includegraphics[scale=.75]{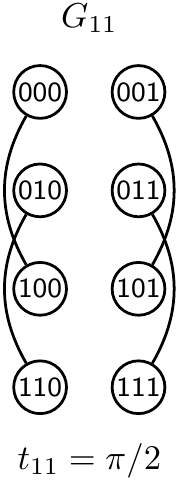} \thinspace
			\includegraphics[scale=.75]{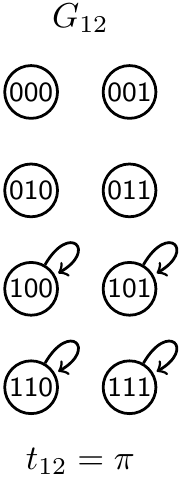} \thinspace
			\includegraphics[scale=.75]{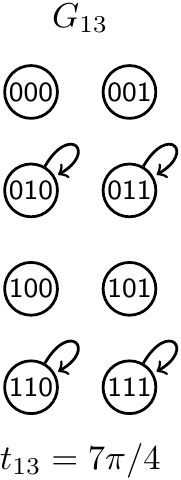} \thinspace
			\includegraphics[scale=.75]{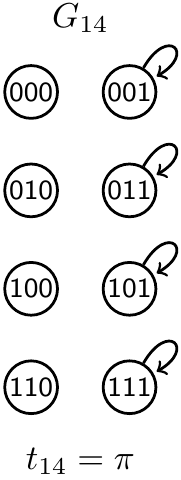} \thinspace
			\includegraphics[scale=.75]{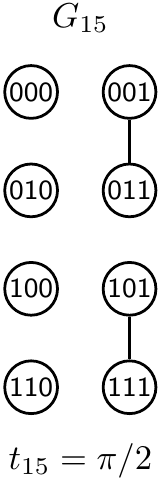} \thinspace
			\includegraphics[scale=.75]{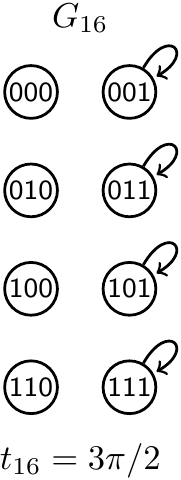} 
		\end{minipage}
	}

	\subfloat[] {
		\label{fig:firstreducedcircuit}
		\begin{minipage}[b]{1.0\linewidth}\centering
			\includegraphics[scale=.75]{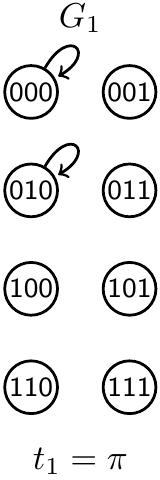} \thinspace
			\includegraphics[scale=.75]{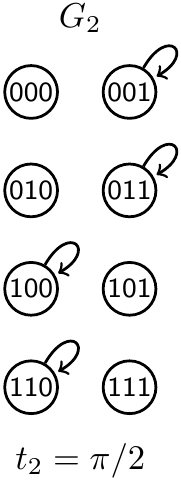} \thinspace
			\includegraphics[scale=.75]{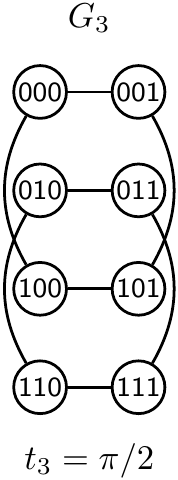} \thinspace
			\includegraphics[scale=.75]{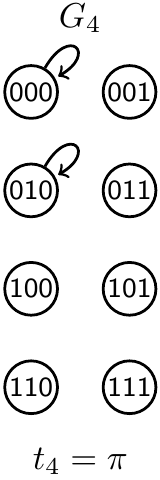} \thinspace
			\includegraphics[scale=.75]{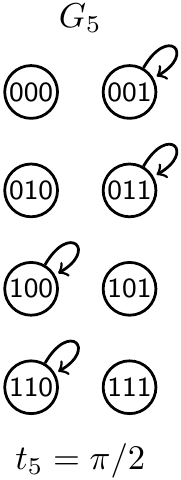} \thinspace
			\includegraphics[scale=.75]{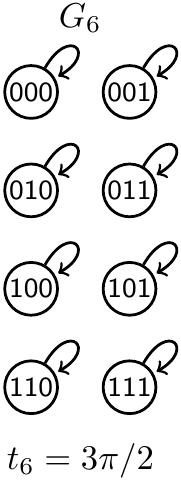} \thinspace
			\includegraphics[scale=.75]{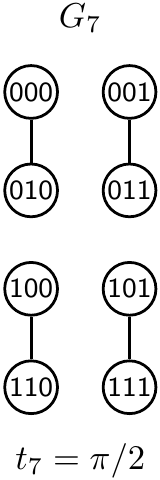} \hfill
			
			\vspace{1em}
			\includegraphics[scale=.75]{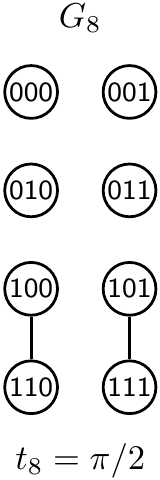} \thinspace
			\includegraphics[scale=.75]{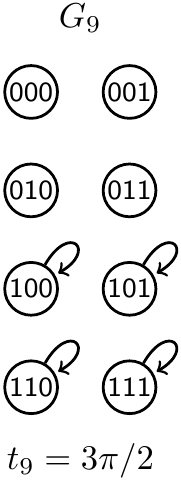} \thinspace
			\includegraphics[scale=.75]{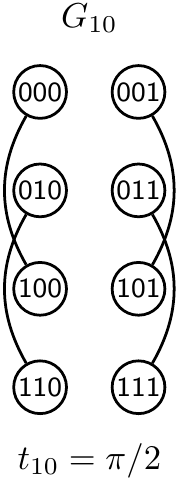} \thinspace
			\includegraphics[scale=.75]{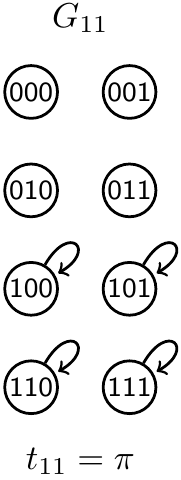} \thinspace
			\includegraphics[scale=.75]{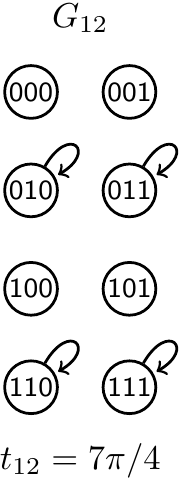} \thinspace
			\includegraphics[scale=.75]{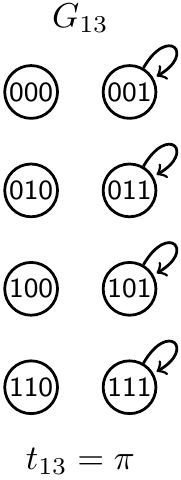} \thinspace
			\includegraphics[scale=.75]{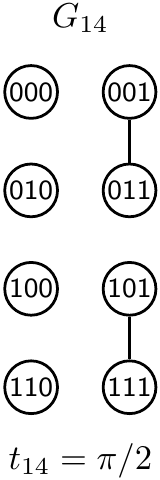} \thinspace
			\includegraphics[scale=.75]{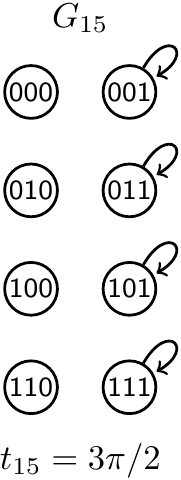} 
		\end{minipage}
	}
	\caption{\label{fig:motivation}(a) The dynamic quantum walk for the circuit pictured in \fref{fig:sequential}. (b) The dynamic graph starting with $H \otimes I \otimes H$ followed by $I \otimes X \otimes I$, then the rest of the operations in \fref{fig:sequential}. Continued on the next page.}
\end{center}
\end{figure*}

Now, let us simplify this dynamic graph using our observations. First, note from \fref{fig:layers} that we can implement the first layer of gates by applying $H \otimes I \otimes H$ followed by $I \otimes X \otimes I$. We can implement $H \otimes I \otimes H$ using the hypercube observation, which is Obs.~\ref{obs:hypercube}, and this will require five graphs. Then, we can apply $I \otimes X \otimes I$, which will take two graphs. Together, these seven graphs replace the first eight graphs of \fref{fig:fullcircuit}, and the result is shown in \fref{fig:firstreducedcircuit}. Let us explore this in more detail. $G_1$ and $G_2$ of \fref{fig:firstreducedcircuit} apply phases in preparation for walking on the hypercube. $G_3$ is a walk on two-dimensional hypercubes, which are cycles of length 4, since the Hadamard gate is applied to two vertices. One $C_4$ connects vertices $\ket{000},\ket{001}, \ket{101}$, and $\ket{100}$, while another $C_4$ connects the other four vertices in a similar cycle. Following, $G_4$ and $G_5$ repeat the phases in $G_1$ and $G_2$, and this finishes $H \otimes I \otimes H$ according to Obs.~\ref{obs:hypercube}. Next, $G_6$ and $G_7$ implement $I \otimes X \otimes I$, and they are the same as $G_5$ and $G_4$ from \fref{fig:fullcircuit}, respectively, so we have swapped their order using Obs.~\ref{obs:commute}. Finally, $G_8$ through $G_{15}$ are the same as $G_9$ through $G_{16}$ from \fref{fig:fullcircuit}.

\begin{figure*}
\begin{center}
	\ContinuedFloat

	\subfloat[] {
		\label{fig:secondreducedcircuit}
		\begin{minipage}[b]{1.0\linewidth}\centering
			\includegraphics[scale=.75]{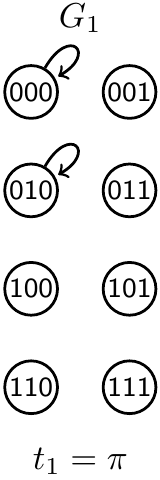} \thinspace
			\includegraphics[scale=.75]{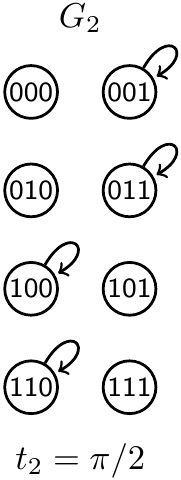} \thinspace
			\includegraphics[scale=.75]{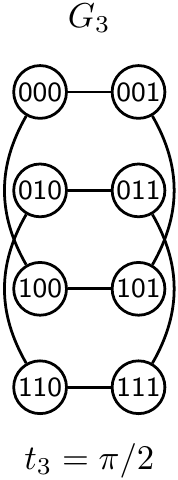} \thinspace
			\includegraphics[scale=.75]{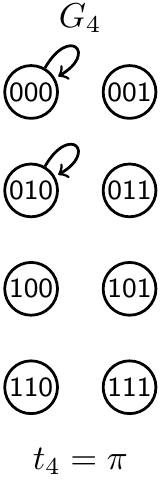} \thinspace
			\includegraphics[scale=.75]{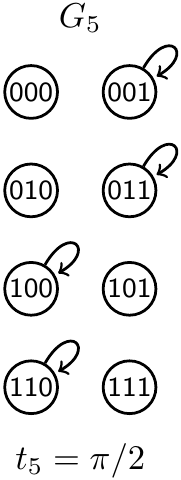} \thinspace
			\includegraphics[scale=.75]{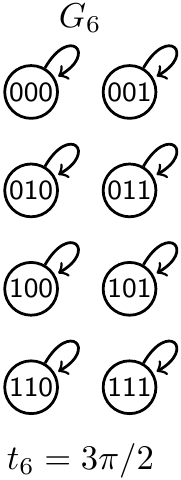} \thinspace
			\includegraphics[scale=.75]{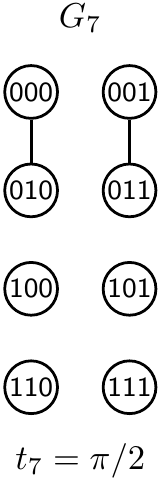} \hfill
			
			\vspace{1em}
			\includegraphics[scale=.75]{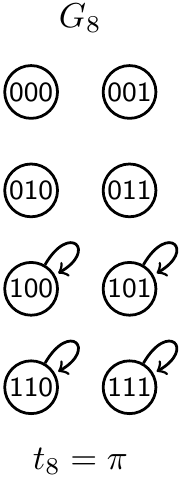} \thinspace
			\includegraphics[scale=.75]{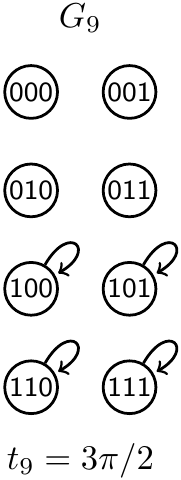} \thinspace
			\includegraphics[scale=.75]{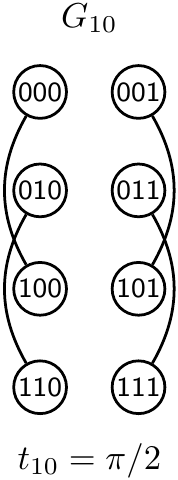} \thinspace
			\includegraphics[scale=.75]{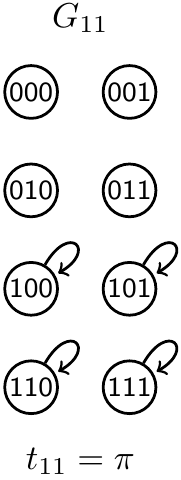} \thinspace
			\includegraphics[scale=.75]{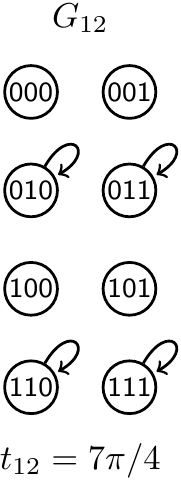} \thinspace
			\includegraphics[scale=.75]{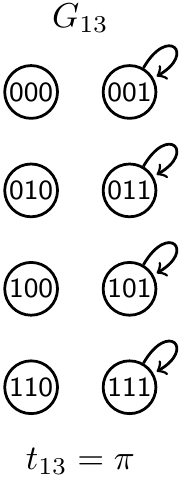} \thinspace
			\includegraphics[scale=.75]{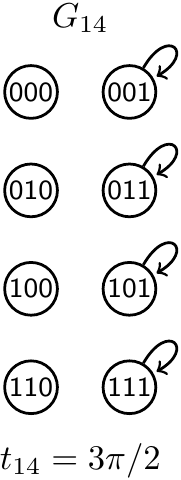} \thinspace
			\includegraphics[scale=.75]{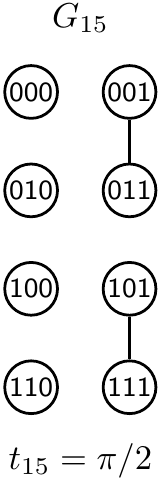} 
		\end{minipage}
	}
	
	\subfloat[] {
		\label{fig:reducedcircuit}
		\begin{minipage}[b]{1.0\linewidth}\centering
			\includegraphics[scale=.75]{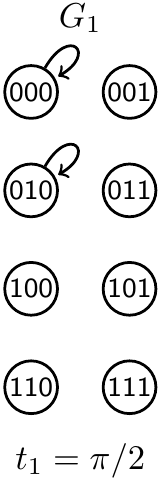} \thinspace
			\includegraphics[scale=.75]{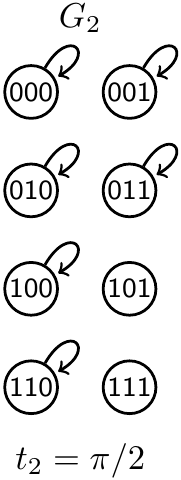} \thinspace
			\includegraphics[scale=.75]{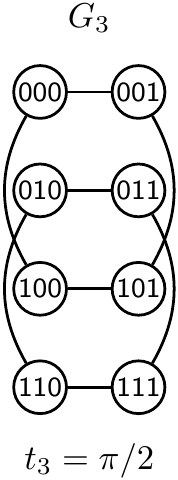} \thinspace
			\includegraphics[scale=.75]{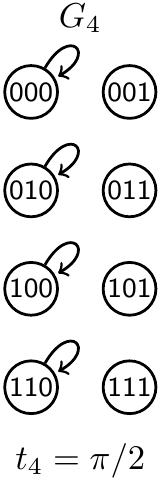} \thinspace
			\includegraphics[scale=.75]{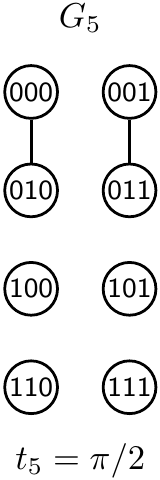} \thinspace
			\includegraphics[scale=.75]{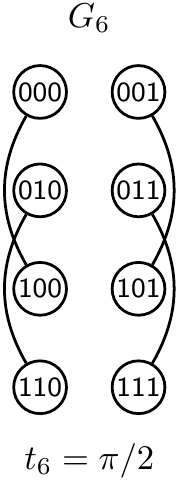} \thinspace
			\includegraphics[scale=.75]{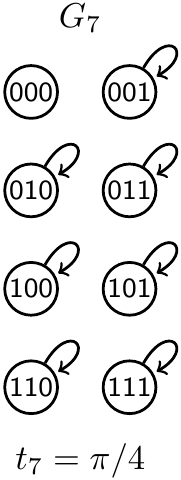} \hfill
			
			\vspace{1em}
			\includegraphics[scale=.75]{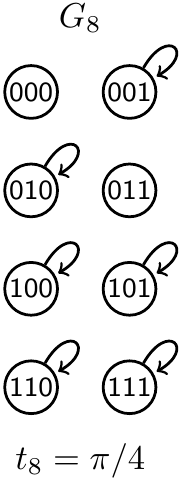} \thinspace
			\includegraphics[scale=.75]{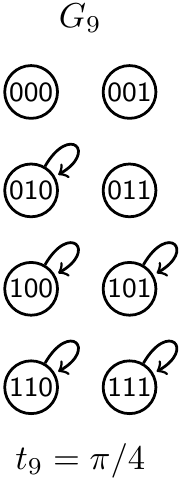} \thinspace
			\includegraphics[scale=.75]{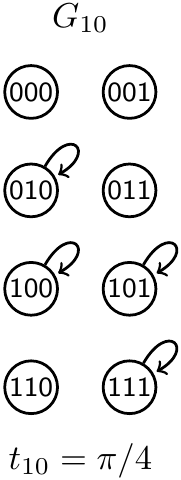} \thinspace
			\includegraphics[scale=.75]{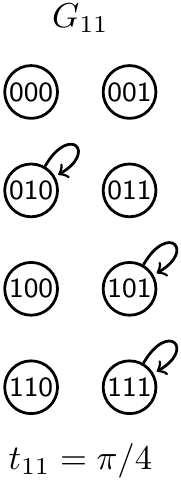} \thinspace
			\includegraphics[scale=.75]{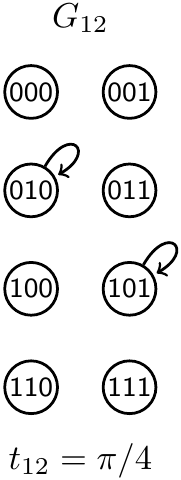} \thinspace
			\includegraphics[scale=.75]{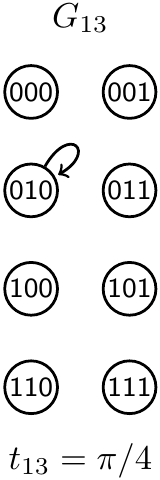} \thinspace
			\includegraphics[scale=.75]{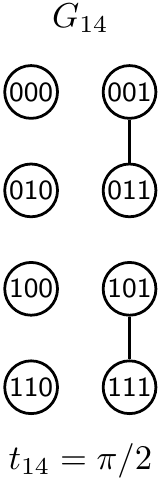} 
		\end{minipage}
	}
	\caption{Continued from the previous page. (c) The dynamic graph after combining identical graphs and swapping commuting graphs. (d) The dynamic graph after combining singletons.}
\end{center}
\end{figure*}

In our next step of simplification, we take the bottom two edges from $G_7$ in Fig.~\ref{fig:firstreducedcircuit}, remove them from $G_7$, and add them to $G_8$ using Obs.~\ref{obs:same}, which increases the time $t_8$ from $\pi/2$ to $\pi$. Since $P_2$ at $t=\pi$ simply applies negative signs, we can instead evolve the vertices as looped singletons for time $t = \pi$, which also applies negative signs. This is shown in \fref{fig:secondreducedcircuit}, where $G_8$ has self-loops on the bottom four vertices. We also swapped $G_{14}$ and $G_{15}$ using Obs.~\ref{obs:commute}.

In order to obtain the final reduced circuit, which is Fig.~\ref{fig:reducedcircuit}, we combine singletons in several graphs. We use Obs.~\ref{obs:singletons} to partially combine the looped singletons from $G_1$ and $G_2$ from \fref{fig:secondreducedcircuit} to get $G_1$ and $G_2$ in \fref{fig:reducedcircuit}. Note we still need two graphs since $t_1 \ne t_2$, however, this reduces the evolution time. Then, note in \fref{fig:secondreducedcircuit} that $G_8$ and $G_9$ commute with $G_7$, so we can swap them by Obs.~\ref{obs:commute}. Then, we can combine these two graphs with $G_4$, $G_5$, and $G_6$ of \fref{fig:secondreducedcircuit} using Obs.~\ref{obs:singletons}, resulting in $G_4$ of \fref{fig:reducedcircuit}. Finally, we combine $G_{12}$, $G_{13}$, and $G_{14}$, and the fastest method of implementing all of these singletons is to allow them to accumulate phase in multiples of $\pi/4$, which is done in graphs $G_7$ through $G_{13}$ in Fig.~\ref{fig:reducedcircuit}. This new implementation takes two fewer graphs and has a total evolution time of $21\pi/4 \approx 16.5$, as opposed to $67\pi/4 \approx 52.6$ for the previous walk, which is roughly a 68.7\% speedup. In Appendix \ref{proofappendix}, we explicitly prove that \fref{fig:fullcircuit} and \fref{fig:reducedcircuit} perform the same computation.


\section{Conclusion}\label{sec:conclusion}

Continuous-time quantum walks on dynamic graphs are universal for quantum computation \cite{HH2019}, so any quantum circuit can be implemented as a continuous-time quantum walk on a dynamic graph. In this paper, we showed how to simplify these dynamic graphs, resulting in dynamic graphs with either fewer graphs in the sequence, a shorter overall evolution time, or both. These simplifications were based on sequential graphs that commute, sequential graphs that are the same, sequential perfect state transfers, sequential graphs that are complementary subgraphs, graphs containing singleton vertices, and uniform mixing on the hypercube. We also showed that in the previous formulation of dynamic quantum walks, setting the Hamiltonian $H = A$ may not be the best way to measure time and thus used the convention that $H = A/\norm{A}$.

Regarding the questions raised in the introduction, we answered positively. There are ways to apply single-qubit gates in parallel, to various degrees, using quantum walks on dynamic graphs, and more broadly, there are properties that allow quantum walks on dynamic graphs to be simplified in certain ways.

Further research includes determining if there are more properties that can be used to simplify dynamic quantum walks. In fact, we believe the conditions under which Obs.~\ref{obs:comp} can be relaxed in some cases. For example, let us consider $X \otimes X$. Note that in the implementation of $X \otimes I$, the edges used connect $\ket{00}$ to $\ket{01}$ and $\ket{10}$ to $\ket{11}$. The edges used in $I \otimes X$ connect $\ket{00}$ to $\ket{10}$ and $\ket{01}$ to $\ket{11}$. Instead of using these graphs in sequence, we can combine them into one $C_4$ graph and propagate it for $\pi$ followed by singletons to change the phase. This case is slightly different from the perfect state transfer observation since the edges in $I \otimes X$ are incident to vertices that are incident to edges in $X \otimes I$.


\begin{acknowledgements}
	R.H.~was supported by DARPA ONISQ program under award W911NF-20-2-0051. The authors thank the organizers of the ``Quantum Information on Graphs'' session of the 2019 Canadian Mathematical Society Winter Meeting, where their collaboration on this research was initiated.
\end{acknowledgements}




\section*{Declarations}

\noindent {\small \textbf{Conflict of interest} \enspace R.H.~ has no competing interests to declare that are relevant to the content of this article. T.W.~is on the Editorial Board of the journal.}

\noindent {\small \textbf{Data and code} \enspace Data sharing not applicable to this article as no datasets were generated or analysed during the current study.}

\appendix

\section{Proof of Equivalent Dynamic Quantum Walks for \fref{fig:motivation}}\label{proofappendix}

Note that the sequence in \fref{fig:fullcircuit} acts on the initial state $\ket{\psi(0)} = c_0 \ket{000} + \dots + c_7 \ket{111}$ via
\begin{align*}
	&c_0 \ket{000} + c_1 \ket{001} + c_2 \ket{010} + c_3 \ket{011} + c_4 \ket{100} + c_5 \ket{101} + c_6 \ket{110} + c_7 \ket{111} \\
	&\quad \xrightarrow{G_1} c_0 \ket{000} + c_1 \ket{001} + c_2\ket{010} + c_3\ket{011} + i \big[c_4 \ket{100} +c_5 \ket{101} + c_6 \ket{110} + c_7 \ket{111}\big]\\
	&\quad \xrightarrow{G_2} 1/\sqrt{2}\big[(c_0+c_4) \ket{000} + (c_1+c_5) \ket{001} + (c_2+c_6) \ket{010} + (c_3+c_7) \ket{011} \\
	&\quad\quad\quad + i \big((c_4-c_0) \ket{100} +(c_5-c_1) \ket{101} + (c_6-c_2) \ket{110} + (c_7-c_3) \ket{111}\big)\big]\\
	&\quad \xrightarrow{G_3} 1/\sqrt{2}\big[(c_0+c_4) \ket{000} + (c_1+c_5) \ket{001} + (c_2+c_6) \ket{010} + (c_3+c_7) \ket{011} \\
	&\quad\quad\quad + (c_0-c_4) \ket{100} +(c_1-c_5) \ket{101} + (c_2-c_6) \ket{110} + (c_3-c_7) \ket{111}\big]\\
	&\quad \xrightarrow{G_4} i/\sqrt{2}\big[-\big((c_2+c_6) \ket{000} + (c_3+c_7) \ket{001} + (c_0+c_4) \ket{010} + (c_1+c_5) \ket{011}\big) \\
	&\quad\quad\quad + (-c_2+c_6) \ket{100} + (-c_3+c_7) \ket{101} + (-c_0+c_4) \ket{110} + (-c_1+c_5) \ket{111}\big]\\
	&\quad \xrightarrow{G_5} 1/\sqrt{2}\big[(c_2+c_6) \ket{000} + (c_3+c_7) \ket{001} + (c_0+c_4) \ket{010} + (c_1+c_5) \ket{011} \\
	&\quad\quad\quad + (c_2-c_6) \ket{100} +(c_3-c_7) \ket{101} + (c_0-c_4) \ket{110} + (c_1-c_5) \ket{111}\big]\\
	&\quad \xrightarrow{G_6} 1/\sqrt{2}\big[(c_2+c_6) \ket{000} + i(c_3+c_7) \ket{001} + (c_0+c_4) \ket{010} + i(c_1+c_5) \ket{011} \\
	&\quad\quad\quad + (c_2-c_6) \ket{100} +i(c_3-c_7) \ket{101} + (c_0-c_4) \ket{110} + i(c_1-c_5) \ket{111}\big]\\
	&\quad \xrightarrow{G_7} 1/2\big[(c_2+c_3+c_6+c_7) \ket{000} + i(-c_2+c_3-c_6+c_7) \ket{001} \\
	&\quad\quad\quad + (c_0+c_1+c_4+c_5) \ket{010} + i(-c_0+c_1-c_4+c_5) \ket{011} \\
	&\quad\quad\quad + (c_2+c_3-c_6-c_7) \ket{100} +i(-c_2+c_3+c_6-c_7) \ket{101} \\
	&\quad\quad\quad + (c_0+c_1-c_4-c_5) \ket{110} + i(-c_0+c_1+c_4-c_5) \ket{111}\big]\\
	&\quad \xrightarrow{G_8} 1/2\big[(c_2+c_3+c_6+c_7) \ket{000} + (c_2-c_3+c_6-c_7) \ket{001} \\
	&\quad\quad\quad + (c_0+c_1+c_4+c_5) \ket{010} + (c_0-c_1+c_4-c_5) \ket{011} \\
	&\quad\quad\quad + (c_2+c_3-c_6-c_7) \ket{100} +(c_2-c_3-c_6+c_7) \ket{101} \\
	&\quad\quad\quad + (c_0+c_1-c_4-c_5) \ket{110} + (c_0-c_1-c_4+c_5) \ket{111}\big]\\
	&\quad \xrightarrow{G_9} 1/2\big[(c_2+c_3+c_6+c_7) \ket{000} + (c_2-c_3+c_6-c_7) \ket{001} \\
	&\quad\quad\quad + (c_0+c_1+c_4+c_5) \ket{010} + (c_0-c_1+c_4-c_5) \ket{011} \\
	&\quad\quad\quad - i\big((c_0+c_1-c_4-c_5) \ket{100} +(c_0-c_1-c_4+c_5) \ket{101} \\
	&\quad\quad\quad + (c_2+c_3-c_6-c_7) \ket{110} + (c_2-c_3-c_6+c_7) \ket{111}\big)\big]\\
	&\quad \xrightarrow{G_{10}} 1/2\big[(c_2+c_3+c_6+c_7) \ket{000} + (c_2-c_3+c_6-c_7) \ket{001} \\
	&\quad\quad\quad + (c_0+c_1+c_4+c_5) \ket{010} + (c_0-c_1+c_4-c_5) \ket{011} \\
	&\quad\quad\quad + (c_0+c_1-c_4-c_5) \ket{100} +(c_0-c_1-c_4+c_5) \ket{101} \\
	&\quad\quad\quad + (c_2+c_3-c_6-c_7) \ket{110} + (c_2-c_3-c_6+c_7) \ket{111}\big]\\
	&\quad \xrightarrow{G_{11}} -i/2\big[(c_0+c_1-c_4-c_5) \ket{000} + (c_0-c_1-c_4+c_5) \ket{001} \\
	&\quad\quad\quad + (c_2+c_3-c_6-c_7) \ket{010} + (c_2-c_3-c_6+c_7) \ket{011} \\
	&\quad\quad\quad + (c_2+c_3+c_6+c_7) \ket{100} +(c_2-c_3+c_6-c_7) \ket{101} \\
	&\quad\quad\quad + (c_0+c_1+c_4+c_5) \ket{110} + (c_0-c_1+c_4-c_5) \ket{111}\big]\\
	&\quad \xrightarrow{G_{12}} -i/2\big[(c_0+c_1-c_4-c_5) \ket{000} + (c_0-c_1-c_4+c_5) \ket{001} \\
	&\quad\quad\quad + (c_2+c_3-c_6-c_7) \ket{010} + (c_2-c_3-c_6+c_7) \ket{011} \\
	&\quad\quad\quad - (c_2+c_3+c_6+c_7) \ket{100} -(c_2-c_3+c_6-c_7) \ket{101} \\
	&\quad\quad\quad - (c_0+c_1+c_4+c_5) \ket{110} - (c_0-c_1+c_4-c_5) \ket{111}\big]\\
	&\quad \xrightarrow{G_{13}} -i/2\big[(c_0+c_1-c_4-c_5) \ket{000} + (c_0-c_1-c_4+c_5) \ket{001} \\
	&\quad\quad\quad + e^{i\pi/4}(c_2+c_3-c_6-c_7) \ket{010} + e^{i\pi/4}(c_2-c_3-c_6+c_7) \ket{011} \\
	&\quad\quad\quad - (c_2+c_3+c_6+c_7) \ket{100} -(c_2-c_3+c_6-c_7) \ket{101} \\
	&\quad\quad\quad - e^{i\pi/4}(c_0+c_1+c_4+c_5) \ket{110} - e^{i\pi/4}(c_0-c_1+c_4-c_5) \ket{111}\big]\\
	&\quad \xrightarrow{G_{14}} -i/2\big[(c_0+c_1-c_4-c_5) \ket{000} - (c_0-c_1-c_4+c_5) \ket{001} \\
	&\quad\quad\quad + e^{i\pi/4}(c_2+c_3-c_6-c_7) \ket{010} - e^{i\pi/4}(c_2-c_3-c_6+c_7) \ket{011} \\
	&\quad\quad\quad - (c_2+c_3+c_6+c_7) \ket{100} -(c_2-c_3+c_6-c_7) \ket{101} \\
	&\quad\quad\quad - e^{i\pi/4}(c_0+c_1+c_4+c_5) \ket{110} + e^{i\pi/4}(c_0-c_1+c_4-c_5) \ket{111}\big]\\
	&\quad \xrightarrow{G_{15}} 1/2\big[i(-c_0-c_1+c_4+c_5) \ket{000} + e^{i\pi/4}(c_2-c_3-c_6+c_7) \ket{001} \\
	&\quad\quad\quad - i e^{i\pi/4}(c_2+c_3-c_6-c_7) \ket{010} + (c_0-c_1-c_4+c_5) \ket{011} \\
	&\quad\quad\quad - i(c_2+c_3+c_6+c_7) \ket{100} + e^{i\pi/4}(-c_0+c_1-c_4+c_5) \ket{101} \\
	&\quad\quad\quad + i e^{i\pi/4}(c_0+c_1+c_4+c_5) \ket{110} -(c_2-c_3+c_6-c_7) \ket{111}\big]\\
	&\quad \xrightarrow{G_{16}} i/2\big[(-c_0-c_1+c_4+c_5) \ket{000} + e^{i\pi/4}(c_2-c_3-c_6+c_7) \ket{001} \\
	&\quad\quad\quad - e^{i\pi/4}(c_2+c_3-c_6-c_7) \ket{010} + (c_0-c_1-c_4+c_5) \ket{011} \\
	&\quad\quad\quad - (c_2+c_3+c_6+c_7) \ket{100} +e^{i\pi/4}(-c_0+c_1-c_4+c_5) \ket{101} \\
	&\quad\quad\quad + e^{i\pi/4}(c_0+c_1+c_4+c_5) \ket{110} -(c_2-c_3+c_6-c_7) \ket{111}\big].
\end{align*}
The dynamic graph in \fref{fig:reducedcircuit} acts on the initial state $\ket{\psi(0)} = c_0 \ket{000} + \dots + c_7 \ket{111}$ via 
\begin{align*}
	&c_0 \ket{000} + c_1 \ket{001} + c_2 \ket{010} + c_3 \ket{011} + c_4 \ket{100} + c_5 \ket{101} + c_6 \ket{110} + c_7 \ket{111}\\
	&\quad \xrightarrow{G_1} -ic_0 \ket{000} + c_1 \ket{001} - ic_2\ket{010} + c_3\ket{011} + c_4 \ket{100} + c_5 \ket{101} + c_6 \ket{110} + c_7 \ket{111}\\
	&\quad \xrightarrow{G_2} -c_0 \ket{000} -ic_1 \ket{001} -c_2\ket{010} -ic_3\ket{011} - ic_4 \ket{100} +c_5 \ket{101} - ic_6 \ket{110} + c_7 \ket{111}\\
	&\quad \xrightarrow{G_3} 1/2\big[-(c_0+c_1+c_4+c_5) \ket{000} + i(c_0-c_1+c_4-c_5) \ket{001} \\
	&\quad\quad\quad - (c_2+c_3+c_6+c_7) \ket{010} + i(c_2-c_3+c_6-c_7) \ket{011} \\
	&\quad\quad\quad + i(c_0+c_1-c_4-c_5) \ket{100} +(c_0-c_1-c_4+c_5) \ket{101} \\
	&\quad\quad\quad + i(c_2+c_3-c_6-c_7) \ket{110} + (c_2-c_3-c_6+c_7)\ket{111}\big]\\
	&\quad \xrightarrow{G_4} 1/2\big[i(c_0+c_1+c_4+c_5) \ket{000} + i(c_0-c_1+c_4-c_5) \ket{001} \\
	&\quad\quad\quad + i(c_2+c_3+c_6+c_7) \ket{010} + i(c_2-c_3+c_6-c_7) \ket{011} \\
	&\quad\quad\quad + (c_0+c_1-c_4-c_5) \ket{100} +(c_0-c_1-c_4+c_5) \ket{101} \\
	&\quad\quad\quad + (c_2+c_3-c_6-c_7) \ket{110} + (c_2-c_3-c_6+c_7)\ket{111}\big]\\
	&\quad \xrightarrow{G_5} 1/2\big[(c_2+c_3+c_6+c_7) \ket{000} + (c_2-c_3+c_6-c_7) \ket{001} \\
	&\quad\quad\quad + (c_0+c_1+c_4+c_5) \ket{010} + (c_0-c_1+c_4-c_5) \ket{011} \\
	&\quad\quad\quad + (c_0+c_1-c_4-c_5) \ket{100} +(c_0-c_1-c_4+c_5) \ket{101} \\
	&\quad\quad\quad + (c_2+c_3-c_6-c_7) \ket{110} + (c_2-c_3-c_6+c_7)\ket{111}\big]\\
	&\quad \xrightarrow{G_6} i/2\big[(-c_0-c_1+c_4+c_5) \ket{000} + (-c_0+c_1+c_4-c_5) \ket{001} \\
	&\quad\quad\quad + (-c_2-c_3+c_6+c_7) \ket{010} + (-c_2+c_3+c_6-c_7) \ket{011} \\
	&\quad\quad\quad - (c_2+c_3+c_6+c_7) \ket{100} +(-c_2+c_3-c_6+c_7) \ket{101} \\
	&\quad\quad\quad - (c_0+c_1+c_4+c_5) \ket{110} + (-c_0+c_1-c_4+c_5)\ket{111}\big]\\
	&\quad \xrightarrow{G_7} i/2\big\{(-c_0-c_1+c_4+c_5) \ket{000} + e^{-i\pi/4}\big[ (-c_0+c_1+c_4-c_5) \ket{001} \\
	&\quad\quad\quad + (-c_2-c_3+c_6+c_7) \ket{010} + (-c_2+c_3+c_6-c_7) \ket{011} \\
	&\quad\quad\quad - (c_2+c_3+c_6+c_7) \ket{100} + (-c_2+c_3-c_6+c_7) \ket{101} \\
	&\quad\quad\quad - (c_0+c_1+c_4+c_5) \ket{110} + (-c_0+c_1-c_4+c_5)\ket{111} \big]\big\}\\
	&\quad \xrightarrow{G_8} 1/2\big[i(-c_0-c_1+c_4+c_5) \ket{000} + (-c_0+c_1+c_4-c_5) \ket{001} \\ 
	&\quad\quad\quad + (-c_2-c_3+c_6+c_7) \ket{010} + i e^{-i\pi/4}(-c_2+c_3+c_6-c_7) \ket{011} \\
	&\quad\quad\quad - (c_2+c_3+c_6+c_7) \ket{100} +(-c_2+c_3-c_6+c_7) \ket{101} \\
	&\quad\quad\quad - (c_0+c_1+c_4+c_5) \ket{110} + (-c_0+c_1-c_4+c_5)\ket{111}\big]\\
	&\quad \xrightarrow{G_9} 1/2\big\{i(-c_0-c_1+c_4+c_5) \ket{000} + (-c_0+c_1+c_4-c_5) \ket{001} \\
	&\quad\quad\quad + e^{-i\pi/4} \big[ (-c_2-c_3+c_6+c_7) \ket{010} + i(-c_2+c_3+c_6-c_7) \ket{011} \\
	&\quad\quad\quad - (c_2+c_3+c_6+c_7) \ket{100} +(-c_2+c_3-c_6+c_7) \ket{101}\\
	&\quad\quad\quad - (c_0+c_1+c_4+c_5) \ket{110} + (-c_0+c_1-c_4+c_5)\ket{111}\big]\big\}\\
	&\quad \xrightarrow{G_{10}} 1/2\big[i(-c_0-c_1+c_4+c_5) \ket{000} + (-c_0+c_1+c_4-c_5) \ket{001} \\
	&\quad\quad\quad + i(-c_2-c_3+c_6+c_7) \ket{010} + ie^{-i\pi/4}(-c_2+c_3+c_6-c_7) \ket{011} \\
	&\quad\quad\quad - i(c_2+c_3+c_6+c_7) \ket{100} +i(-c_2+c_3-c_6+c_7) \ket{101} \\
	&\quad\quad\quad - e^{-i\pi/4}(c_0+c_1+c_4+c_5) \ket{110} + i(-c_0+c_1-c_4+c_5)\ket{111}\big]\\
	&\quad \xrightarrow{G_{11}} 1/2\big[i(-c_0-c_1+c_4+c_5) \ket{000} + (-c_0+c_1+c_4-c_5) \ket{001}\\
	&\quad\quad\quad + ie^{-i\pi/4}(c_2+c_3-c_6-c_7) \ket{010} + ie^{-i\pi/4}(-c_2+c_3+c_6-c_7) \ket{011} \\
	&\quad\quad\quad + i(c_2+c_3+c_6+c_7) \ket{100} +ie^{-i\pi/4}(c_2-c_3+c_6-c_7) \ket{101} \\
	&\quad\quad\quad - e^{-i\pi/4}(c_0+c_1+c_4+c_5) \ket{110} + ie^{-i\pi/4}(c_0-c_1+c_4-c_5)\ket{111}\big]\\
	&\quad \xrightarrow{G_{12}} 1/2\big[i(-c_0-c_1+c_4+c_5) \ket{000} + (-c_0+c_1+c_4-c_5) \ket{001} \\
	&\quad\quad\quad + (c_2+c_3-c_6-c_7) \ket{010} + ie^{-i\pi/4}(-c_2+c_3+c_6-c_7) \ket{011} \\
	&\quad\quad\quad + i(c_2+c_3+c_6+c_7) \ket{100} +(c_2-c_3+c_6-c_7) \ket{101} \\
	&\quad\quad\quad - e^{-i\pi/4}(c_0+c_1+c_4+c_5) \ket{110} + ie^{-i\pi/4}(c_0-c_1+c_4-c_5)\ket{111}\big]\\
	&\quad \xrightarrow{G_{13}} 1/2\big[i(-c_0-c_1+c_4+c_5) \ket{000} + (-c_0+c_1+c_4-c_5) \ket{001} \\
	&\quad\quad\quad + e^{-i\pi/4}(c_2+c_3-c_6-c_7) \ket{010} + ie^{-i\pi/4}(-c_2+c_3+c_6-c_7) \ket{011} \\
	&\quad\quad\quad + i(c_2+c_3+c_6+c_7) \ket{100} +(c_2-c_3+c_6-c_7) \ket{101} \\
	&\quad\quad\quad - e^{-i\pi/4}(c_0+c_1+c_4+c_5) \ket{110} + ie^{-i\pi/4}(c_0-c_1+c_4-c_5)\ket{111}\big]\\
	&\quad \xrightarrow{G_{14}} 1/2\big[i(-c_0-c_1+c_4+c_5) \ket{000} + e^{-i\pi/4}(-c_2+c_3+c_6-c_7) \ket{001} \\
	&\quad\quad\quad + e^{-i\pi/4}(c_2+c_3-c_6-c_7) \ket{010} + i(c_0-c_1-c_4+c_5) \ket{011} \\
	&\quad\quad\quad + i(c_2+c_3+c_6+c_7) \ket{100} +e^{-i\pi/4}(c_0-c_1+c_4-c_5) \ket{101} \\
	&\quad\quad\quad - e^{-i\pi/4}(c_0+c_1+c_4+c_5) \ket{110} + i(-c_2+c_3-c_6+c_7)\ket{111}\big].
\end{align*}
The two final states are identical, since $ie^{-i\pi/4} = e^{i\pi/4}$, hence the dynamic quantum walks are equivalent.


\bibliographystyle{qinp}
\bibliography{refs}

\end{document}